\documentclass[aps,pra,one column,nofootinbib,superscriptaddress,10pt,floatfix,longbibliography]{revtex4-2}

\usepackage{minitoc}
\usepackage[toc,page,header]{appendix}
\usepackage{physics}
\usepackage{xfrac}
\usepackage{graphicx}
\usepackage{subfigure}
\usepackage{mathdots}
\usepackage{mathtools}
\usepackage{amsfonts,amssymb,amsmath}
\usepackage{bm}
\usepackage{mathrsfs}
\usepackage[]{graphics,graphicx,epsfig}
\usepackage{amsthm}
\usepackage{csquotes}
\MakeOuterQuote{"}
\usepackage{epstopdf}
\usepackage{tikz}
\usepackage{paralist}
\usepackage{diagbox}
\usepackage[inline]{enumitem}
\usepackage{qcircuit}
\usepackage{dsfont}
\usepackage{hyperref}
\usepackage{algorithm}
\usepackage{algpseudocode}
\usepackage{ninecolors}
\usepackage[normalem]{ulem}

\def\identity{\leavevmode\hbox{\small1\kern-3.8pt\normalsize1}}

\newtheorem{theorem}{Theorem}

\newtheorem{lemma}{Lemma}
\newtheorem{proposition}{Proposition}

\newtheorem{corollary}{Corollary}
\theoremstyle{remark}
\newtheorem{remark}[theorem]{Remark}
\newtheorem*{remarks}{Remark}

\newcommand{\caB}{\mathcal{B}}

\newcommand{\caE}{\mathcal{E}}
\newcommand{\caF}{\mathcal{F}}

\newcommand{\caH}{\mathcal{H}}
\newcommand{\caJ}{\mathcal{J}}
\newcommand{\caL}{\mathcal{L}}

\newcommand{\caN}{\mathcal{N}}

\newcommand{\caV}{\mathcal{V}}
\newcommand{\caW}{\mathcal{W}}

\newcommand{\rmi}{\mathrm{i}}

\newcommand{\rmB}{\mathrm{B}}

\def\eqref#1{\textup{(\ref{#1})}}
\newcommand{\eref}[1]{Eq.~\textup{(\ref{#1})}}

\newcommand{\lref}[1]{Lemma~\ref{#1}}

\newcommand{\thref}[1]{Theorem~\ref{#1}}

\newcommand{\coref}[1]{Corollary~\ref{#1}}

\newcommand{\sref}[1]{Sec.~\ref{#1}}

\newcommand{\pref}[1]{Proposition~\ref{#1}}

\newcommand{\Pref}[1]{Proposition~\ref{#1}}

\newcommand{\fref}[1]{Fig.~\ref{#1}}

\newcommand{\aref}[1]{Appendix~\ref{#1}}

\def\<{\langle}  %% overriding the original command \<
\def\>{\rangle}  %% overriding the original command \>

\newcommand{\rmK}{\mathrm{K}}

\newcommand{\rmM}{\mathrm{M}}

\newcommand{\bbC}{\mathbb{C}}
\newcommand{\bbR}{\mathbb{R}}

\newcommand{\caA}{\mathcal{A}}
\newcommand{\caK}{\mathcal{K}}

\newcommand{\rb}{\rho_\beta}
\newcommand{\bbP}{\mathbb{P}}

\begin{document}

\title{Modular-Annihilator Parent Hamiltonians for Purified Gibbs States: Spectral Design and Controlled Approximation}

\author{Changhao Yi}
\affiliation{Department of Physics, Shanghai University, Shanghai, China}
\email{cyi@shu.edu.cn}

\author{Jun Takahashi}
\affiliation{Institute of Solid State Physics, University of Tokyo, Chiba, Japan}

\author{Cunlu Zhou}
\affiliation{Department of Computer Science, Universit\'e de Sherbrooke, QC, Canada}
\affiliation{Institut quantique, Universit\'e de Sherbrooke, QC, Canada}
\email{Cunlu.Zhou@USherbrooke.ca}

\begin{abstract}
Purified Gibbs states provide a bridge between finite-temperature physics, dissipative dynamics, and ground-state methods. In this work, we study the exact finite sum-of-squares (SoS) construction of their parent Hamiltonians and the associated Lindbladian based on modular annihilators. Given a finite set of Hermitian generators, the corresponding modular annihilators yield a frustration-free SoS representation without continuous time integrals or an explicit decomposition into Bohr-frequency sectors. The purified Gibbs state remains a common zero mode while the freedom to choose and combine the generators can be used to optimize the spectral properties of the parent Hamiltonian. For free-fermion Hamiltonians, modular transformations act linearly on Majorana operators, leading to an analytically solvable family of parent Hamiltonians parameterized by a real symmetric coefficient matrix \(S\). For the scalar-functional subclass $S=f(h)$, we show that, at fixed
operator norm, the choice $S_{\mathrm{opt}}\propto 1/\sqrt{\cosh(2\beta h)}$ has mixing time upper bound $2\log(2N/\epsilon)$ for any $\beta$, which exhibits rapid mixing and is irrelevant to the inverse temperature $\beta$. For interacting systems, where the modularly dressed generators are not available in closed form, we introduce a Krylov--Lanczos approximation scheme and bound the resulting ground-state error in terms of the modular-approximation error and the parent-Hamiltonian gap. Numerical results illustrate the free-fermion spectral advantage and show how the accuracy of the interacting construction depends on temperature, interaction strength, and Krylov dimension.
\end{abstract}

\date{\today}
\maketitle

\tableofcontents

\section{Introduction}

Purified Gibbs states provide a useful bridge between finite-temperature physics, dissipative dynamics, and ground-state methods. For a quantum many-body Hamiltonian \(H\), the Gibbs state
\(\rho_\beta \propto e^{-\beta H}\) encodes equilibrium properties at inverse temperature \(\beta\), including thermal expectation values, phase structure, and response functions
\cite{Poulin2009Sampling,Kastoryano2013Quantum,
Holmes2022QuantumAlgorithms,Hahn2026Efficient,chen2025efficient,
Kastoryano2016Quantum,chen2023an,ding2025efficient,fast2025tong,
Smid2025,rouze2026optimal,rouze2026efficient}.
Its canonical purification, also known as the thermofield-double state \cite{Cottrell2019,khor2026}, represents the mixed Gibbs state as a pure state on a doubled Hilbert space. Constructing a parent Hamiltonian whose ground state is the purified Gibbs state therefore allows finite-temperature questions to be studied using ground-state spectral methods. Furthermore, parent Hamiltonians of purified Gibbs states arise naturally from thermalizing open-system dynamics. Under vectorization and a similarity transformation, a Lindbladian satisfying Kubo--Martin--Schwinger (KMS) detailed balance is mapped to a positive-semidefinite Hermitian operator \cite{Lindblad1976,AlickiLendi2007,Kossakowski1977Quantum,
Scandi2026Thermalization}. The purified Gibbs state is a zero-energy ground state of this operator, and its spectral gap coincides with the corresponding Lindbladian gap. This correspondence connects spectral properties of parent Hamiltonians to both ground-state preparation and dissipative thermalization.

A particularly useful structure is a sum-of-squares (SoS) factorization of the parent Hamiltonian, \(\rmM_{\mathcal H} =\sum_a \Gamma_a^\dagger \Gamma_a/2\). Such a representation makes positivity and frustration freeness explicit and characterizes the ground space as the common kernel of the operators \(\{\Gamma_a\}\). It can also be used by SoS spectral-amplification methods to improve the spectral-gap dependence of ground-state filtering \cite{Low2025Fast,king2026quantum}. These structural and algorithmic advantages motivate the search for explicit SoS representations of thermal parent Hamiltonians. A recent construction obtained such a factorization for parent Hamiltonians associated with a broad class of KMS Lindbladians \cite{leng2026accelerate}. Its natural expression, however, involves a continuous time integral, or equivalently, a decomposition into Bohr-frequency sectors. 
Although this form is well suited to frequency domain analysis, it does not directly expose an exact finite family of algebraically defined annihilators. 

In this work, we study the exact finite SoS construction based on modular transformations, which was first introduced in \cite{Cottrell2019}. Given a finite set of Hermitian generators
\(\{J_a\}\), we define the corresponding modular annihilators by comparing the left and right actions of their modularly dressed forms. Each modular annihilator annihilates the purified Gibbs state, and their squared norms therefore define a frustration-free parent Hamiltonian without continuous time integrals or an explicit enumeration of Bohr frequencies. When the generators form a unital irreducible algebra, their commutant is trivial, which guarantees that the purified Gibbs state is the unique common zero mode. The construction thus reduces the specification of the ground state to an algebraic property of the generators. The finite representation also exposes freedom in the choice and combination of the generators. Invertible changes of generator basis preserve their
common kernel while modifying the positive quadratic form used to construct
the parent Hamiltonian. Consequently, different choices can have the same
target ground state but different excited-state spectra. This makes it
possible to optimize spectral properties such as the gap and operator norm
within analytically or computationally tractable families. 

Free-fermion Hamiltonians provide an exactly solvable realization of this
framework~\cite{Barthel2022,fast2025tong,Smid2025,zhan2026rapid}. Because commutation with a quadratic Hamiltonian preserves the linear span of the Majorana operators, their modular transformations can be evaluated from the single-particle Hamiltonian, which yields a family of SoS parent Hamiltonians parameterized by a real symmetric coefficient matrix \(S\). The Gaussian choice recovers the parent Hamiltonian associated with the Ding--Li--Lin (DLL) Lindbladian~\cite{ding2025efficient}, but it is only one member of this broader family. For the scalar-functional subclass $S=f(h)$, we derive closed-form
expressions for the spectral gap and operator norm and show that
$S_{\mathrm{opt}}\propto 1/\sqrt{\cosh(2\beta h)}$ maximizes the
spectral gap at fixed operator norm within this subclass. The same construction
also defines an exactly solvable family of KMS Lindbladians, allowing its
effect on dissipative relaxation to be analyzed directly. These results are
obtained in the Majorana representation without explicitly invoking the
third-quantization formalism \cite{Prosen2008}.

For interacting Hamiltonians, the modularly dressed generators generally do
not remain in a low-dimensional invariant operator space and are therefore
not available in closed form. The modular-annihilator construction remains
exact, but its practical evaluation requires approximating the
imaginary-time Heisenberg evolution. We address this problem using the
Krylov--Lanczos method \cite{nandy2025quantum} and establish a
stability bound that relates errors in the modularly dressed generators to
the ground-state fidelity of the resulting approximate parent Hamiltonian.
For one-dimensional local Hamiltonians, we combine the Krylov approximation with Araki-type imaginary-time quasi-locality \cite{perez2023locality} and show that at sufficiently high temperature, the modular transformation of a local generator can  be approximated within a finite spatial neighborhood whose size is independent of the total system size.

Our numerical results illustrate both regimes. For free fermions, they show
that the optimized coefficient matrix improves the relevant spectral ratio
and the relaxation of the associated Lindbladian relative to the Gaussian
and identity choices. For interacting fermion models, they show that
increasing the Krylov dimension extends the range of temperatures and
interaction strengths over which the modularly dressed operators, the
purified Gibbs ground state, and the parent-Hamiltonian spectral gap are
accurately reproduced. The loss of accuracy at lower temperature or
stronger interaction is accompanied by increasing modular-approximation
error and increased sensitivity to the parent-Hamiltonian gap.

The remainder of this paper is organized as follows.
Section~\ref{sec:background} reviews modular transformations,
KMS-detailed-balance Lindbladians, parent Hamiltonians, irreducible operator
algebras, and Majorana operators. Section~\ref{sec:main} introduces the
modular-annihilator SoS construction, establishes uniqueness of the purified
Gibbs ground state, and studies the inverse mapping to Lindbladians.
Section~\ref{sec:free_fermion} develops the exactly solvable free-fermion
family and analyzes its spectral optimization. Section~\ref{sec:krylov}
introduces the Krylov--Lanczos approximation and the locality-improved error
bound. Section~\ref{sec:numerical} presents the numerical results, and
Section~\ref{sec:conclusion} concludes.

\section{Background}

\label{sec:background}
\subsection{Modular transformation}

Let $\mathcal{H}$ be a finite-dimensional Hilbert space and let
$\mathcal{B}(\mathcal{H})$ denote the algebra of linear operators on
$\mathcal{H}$. Fix a Hermitian operator $H\in\mathcal{B}(\mathcal{H})$
with spectral decomposition
\begin{equation}
  H=\sum_{E\in\operatorname{spec}(H)} E\Pi_E ,
\end{equation}
where $\Pi_E$ is the orthogonal projection onto the eigenspace of $H$
with eigenvalue $E$.

The {Bohr frequencies} of $H$ are the differences of its
eigenvalues,
\begin{equation}
  \mathcal{F}(H)
  \coloneqq
  \bigl\{E-E' : E,E'\in\operatorname{spec}(H)\bigr\}.
\end{equation}
For $\nu\in\mathcal{F}(H)$ and $X\in\mathcal{B}(\mathcal{H})$, define
the {Bohr-frequency sector}
\begin{equation}
  X_\nu
  \coloneqq
  \sum_{\substack{E,E'\in\operatorname{spec}(H)\\ E-E'=\nu}}
  \Pi_E\, X\, \Pi_{E'} .
\end{equation}
Since $\sum_E \Pi_E= I$, every operator decomposes as
\begin{equation}
  X=\sum_{\nu\in\mathcal{F}(H)} X_\nu ,
  \qquad
  [H,X_\nu]=\nu\,X_\nu ,
\end{equation}
i.e.,\ $X_\nu$ is an eigenoperator of the adjoint action
$\operatorname{ad}_H\coloneqq [H,\cdot\,]$ with eigenvalue $\nu$. This
decomposition is unique and orthogonal with respect to the
Hilbert--Schmidt inner product, and it satisfies
$(X_\nu)^\dagger=(X^\dagger)_{-\nu}$.

For $z\in\mathbb{C}$, define the {modular transformation}
$\Delta_z\colon\mathcal{B}(\mathcal{H})\to\mathcal{B}(\mathcal{H})$ by
\begin{equation}
  \Delta_z(X)
  \coloneqq
  e^{zH}\,X\,e^{-zH}.
\end{equation}
The family $\{\Delta_z\}_{z\in\mathbb{C}}$ is a one-parameter group of
algebra automorphisms, and it
acts diagonally on the Bohr-frequency sectors:
\begin{equation}
  \Delta_z(X_\nu)=e^{z\nu}\,X_\nu ,
  \qquad
  \Delta_z(X)=\sum_{\nu\in\mathcal{F}(H)} e^{z\nu}\,X_\nu .
\end{equation}

\subsection{KMS Lindbladians and parent Hamiltonians}

Given a Hamiltonian $H$ and inverse temperature $\beta$, the Gibbs
state and its purification are
\begin{equation}
  \rho_\beta
  \coloneqq
  \frac{e^{-\beta H}}{\Tr\!\left(e^{-\beta H}\right)},
  \qquad
  |\rho_\beta^{1/2}\rangle
  \coloneqq
  \bigl(\rho_\beta^{1/2}\otimes I\bigr)|\Omega\rangle,
\end{equation}
where $|\Omega\rangle\coloneqq\sum_{i=1}^{d}|ii\rangle$ is the unnormalized
maximally entangled state. Throughout this paper, we assume $\beta > 0$. The pure state $|\rho_\beta^{1/2}\rangle$ is also termed as the thermofield double (TFD) state.

A standard route to preparing $\rho_\beta$ is using Lindbladian dynamics \cite{Lindblad1976}. A
standard Lindbladian takes the form
\begin{equation}
  \mathcal{L}(\cdot)
  =
  -\mathrm{i}[C,\cdot]
  +
  \sum_j
  \Bigl(
    L_j\,(\cdot)\,L_j^\dagger
    -
    \tfrac12\bigl\{L_j^\dagger L_j,(\cdot)\bigr\}
  \Bigr),
\end{equation}where $C$ is a Hermitian operator and $\{L_j\}$ is the set of jump
operators. We say that $\caL$ satisfies the KMS detailed balance
condition with respect to $\rho_\beta$ if
\begin{equation}\label{eq:KMS-DB}
  \caL(\cdot)
  =
  \rho_\beta^{1/2}\,
  \caL^\dagger\!\left[
    \rho_\beta^{-1/2}\,(\cdot)\,\rho_\beta^{-1/2}
  \right]
  \rho_\beta^{1/2},
\end{equation}
Equivalently, $\caL^\dag$ is self-adjoint with respect to the KMS inner
product
$\langle A,B\rangle_{\beta}
 \coloneqq\Tr \bigl(\rho_\beta^{1/2}A^\dagger\rho_\beta^{1/2}B\bigr)$.
 
Since $\caL^\dagger(I)=0$, \eref{eq:KMS-DB} implies
$\caL(\rho_\beta)=0$, and hence
$e^{t\caL}(\rho_\beta)=\rho_\beta$ for all $t\ge0$.
Then $\caL$ is called primitive if
it has a unique stationary state $\rho_\beta$ and
\begin{equation}
\lim_{t\to\infty}e^{t\caL}(\sigma)=\rho_\beta
\end{equation}
for every density operator $\sigma$. The mixing time quantifies how fast the evolution is
\begin{equation}
    t_{\mathrm{mix}}(\varepsilon) \coloneqq \inf\left\{ t : \sup_\rho  \frac{1}{2}\left\|e^{\caL t}[\rho] - \rho_\beta\right\|_1 \le \varepsilon\right\}.
\end{equation}

Under vectorization, $\mathcal{L}$ is mapped to a generally
non-Hermitian matrix $\mathrm{M}_{\mathcal{L}}$. Conjugating by
$\rho_\beta^{1/4}$ on both sides, we get a new superoperator
\begin{equation}\label{eq:superH}
  \mathcal{H}(\cdot)
  \coloneqq
  -\rho_\beta^{-1/4}\,
  \mathcal{L}\!\left[
    \rho_\beta^{1/4}\,(\cdot)\,\rho_\beta^{1/4}
  \right]
  \rho_\beta^{-1/4}.
\end{equation}
The KMS detailed balance condition in \eref{eq:KMS-DB} guarantees that
$\mathcal{H}$ is Hermitian with respect to the Hilbert--Schmidt inner
product, $\mathcal{H}^\dagger=\mathcal{H}$, and stationarity of
$\rho_\beta$ gives $\mathcal{H}(\rho_\beta^{1/2})=0$. Its vectorization
$\mathrm{M}_{\mathcal{H}}$ is therefore a Hermitian matrix, related to
$\mathrm{M}_{\mathcal{L}}$ by the similarity transformation
\begin{equation}
  \mathrm{M}_{\mathcal{H}}
  \coloneqq
  -\bigl(\rho_\beta^{-1/4}\otimes\rho_\beta^{-1/4,\top}\bigr)\,
  \mathrm{M}_{\mathcal{L}}\,
  \bigl(\rho_\beta^{1/4}\otimes\rho_\beta^{1/4,\top}\bigr),
\end{equation}
and annihilates the purified Gibbs state,
$\mathrm{M}_{\mathcal{H}}|\rho_\beta^{1/2}\rangle=0$. In this sense
$\mathrm{M}_{\mathcal{H}}$ is a {parent Hamiltonian} of
$\mathcal{L}$. Because $\caH$ is Hermitian and is related to $-\caL$ by a similarity transformation, its spectrum
is real and nonnegative; the spectrum of a Lindblad generator lies
in the closed left half-plane, while KMS symmetry makes it real.
Hence $\mathrm M_{\caH}\succeq0$. Under primitivity, the spectral gap yields the gap-based mixing-time
bound stated in Appendix~\ref{sec:lindbladian}.

\subsection{Irreducible operator algebras and Majorana operators}

Let $\mathcal W$ be a finite-dimensional Hilbert space.
The unital algebra generated by operators $\{J_a\}_{a\in\caA}$ is denoted by
$\mathrm{Alg}(\{J_a\}_{a\in\caA})$. It is called \emph{irreducible} if the only invariant subspaces are
$\{0\}$ and $\mathcal W$. By Burnside's theorem, irreducibility is equivalent to
$\mathrm{Alg}(\{J_a\}_{a\in\caA})=\mathcal B(\mathcal W)$, the algebra of linear operators on $\caW$. This criterion will guarantee uniqueness of the purified Gibbs state as the ground state of the parent Hamiltonian.

Here is one example of irreducible algebra in the $N$-qubit Hilbert space. Let $\{\omega_a\}_{a=1}^{2N}$ denote Majorana operators satisfying
\begin{equation}
\omega_a^\dagger=\omega_a,
\qquad
\{\omega_a,\omega_b\}=2\delta_{ab}I.
\end{equation}
A free-fermion Hamiltonian takes the form
\begin{equation}
H
=
\sum_{a,b=1}^{2N}
h_{ab}\omega_a\omega_b,
\end{equation}
where $h$ is purely imaginary and Hermitian. The Heisenberg evolution of Majorana operators is exactly solvable:
\begin{equation}
e^{zH}\omega_a e^{-zH}
=
\sum_b
(e^{-4zh})_{ab}\omega_b.
\end{equation}
One concrete realization is the Jordan-Wigner transformation in Majorana form:
\begin{equation}
    \omega_{2a-1}=
    \left(\prod_{b<a}Z_b\right)X_a,
    \qquad
    \omega_{2a}=
    \left(\prod_{b<a}Z_b\right)Y_a .
\end{equation}
The algebra generated by the Majorana operators is irreducible on
$(\mathbb C^2)^{\otimes N}$, since products of Majorana operators generate the entire set of Pauli operators.

Another useful property is the following: suppose $g(h)$ is a function of $h$, and let $G = \sum_{a,b}[g(h)]_{ab}\omega_a \omega_b$. Then $[G,H] = 0$, which implies $\Delta_z(G) = G$ for all $z$.

\section{SoS Factorization from Modular Transformations}
\label{sec:main}

In this section, we present the SoS factorization of parent Hamiltonian using modular transformations. Similar result has been proposed \cite{Cottrell2019}.

\subsection{Exact construction of parent Hamiltonian}

Consider the space of $N$-qubit $(\bbC^2)^{\otimes N}$. Let $H$ be a Hamiltonian on $\caB((\bbC^2)^{\otimes N})$ and $\rb$ its Gibbs state at inverse temperature $\beta$. For any $V\in \caB((\bbC^{2})^{\otimes N})$, define the modular transformation
\begin{equation}
    \Delta(V) \coloneqq \Delta_{\beta/4}\left(V\right) = e^{\beta H/4}V e^{-\beta H/4} = \rb^{-\frac{1}{4}}V \rb^{\frac{1}{4}}.
\end{equation}
Let $\{J_a\}_{a\in\caA}$ be a set of operators in $\caB((\bbC^2)^{\otimes N})$. The purified Gibbs state is a ground state of the following parent Hamiltonian:
\begin{equation}\label{eq:modular_annihilator}
    \mathrm{M}_\caH = \frac{1}{2}\sum_{a\in\caA} \Gamma_a^\dag \Gamma_a,\qquad \Gamma_a \coloneqq \Delta^{-1}(J_a) \otimes I - I \otimes \Delta(J_a)^\top.
\end{equation}
We term each $\Gamma_a$ as a modular annihilator, which satisfies $\Gamma_a|\rb^{1/2}\> = 0$:
\begin{equation}
    \Delta^{-1}(J_a) \rb^{\frac{1}{2}} = \rb^{\frac{1}{2}}\Delta(J_a) \longrightarrow \Delta^{-1}(J_a) \otimes I |\rb^{\frac{1}{2}}\> = I\otimes \Delta(J_a)^\top |\rb^{\frac{1}{2}}\>.
\end{equation}
The full parent Hamiltonian writes
\begin{equation}
    \mathrm{M}_\caH = \frac{1}{2}\sum_{a\in\caA} \left[\Delta(J_a^\dag) \otimes I - I\otimes (\Delta(J_a))^*\right] \left[\Delta^{-1}(J_a) \otimes I - I\otimes (\Delta(J_a))^\top\right].
\end{equation}
Given the parent Hamiltonian, the task of preparing $|\rb^{1/2}\>$ becomes equivalent to preparing the ground state of $\mathrm{M}_\caH$. The next theorem shows when $\mathrm{M}_\caH$ has a unique ground state.
\begin{theorem} \label{thm:algebra}
    Suppose $\mathrm{Alg}(\{J_a\})$ is a unital irreducible algebra on $(\bbC^{2})^{\otimes N}$, and $\{\Gamma_a\}$ is the set of modular annihilators defined in \eref{eq:modular_annihilator}. Then $\Gamma_a|v\> = 0$ for all $a$ iff $|v\> \propto |\rb^{1/2}\>$, and $\rmM_\caH$ has a unique ground state.
\end{theorem}

\begin{proof}
    Let $O_v$ be the matrix form of $|v\>$. According to Schur's lemma in operator algebra, if $\mathrm{Alg}(\{J_a\})$ is irreducible, then $[O,J_a] = 0$ for all $a$ iff $O = cI$. Observe that the condition $\Gamma_a|v\> = 0$ for all $a$ is equivalent to
    \begin{equation}
       e^{-\beta H/4} J_a e^{\beta H/4} O_v  = O_v e^{\beta H/4}J_a e^{-\beta H/4} \Longrightarrow [J_a, e^{\beta H/4} O_v e^{\beta H/4}] = 0 \quad \forall a.
    \end{equation}
    Hence, by the irreducibility of $\mathrm{Alg}(\{J_a\})$, we have $e^{\beta H/4}O_v e^{\beta H/4} \propto I$. Hence, $O_v \propto e^{-\beta H/2}$.
\end{proof}

Other types of SoS parent Hamiltonians can also be derived accordingly. Observe that
\begin{equation}
    e^{-\rmi Ht} \otimes (e^{\rmi H t})^\top|\rb^{\frac{1}{2}}\> = |\rb^{\frac{1}{2}}\>.
\end{equation}
Hence, suppose $\Gamma_a$ is a modular annihilator, then 
\begin{equation}
    \Gamma_a(t) = \left(e^{\rmi Ht} \otimes e^{-\rmi H^\top t}\right) \Gamma_a \left(e^{-\rmi Ht} \otimes e^{\rmi H^\top t}\right)
\end{equation}
is also a modular annihilator such that $\Gamma_a(t)|\rb^{1/2}\> = 0$. This observation gives us another type of parent Hamiltonian which was first mentioned in Ref. \cite{leng2026accelerate}:
\begin{equation}
     \mathrm{M}^{\rm int}_\caH = \frac{1}{2}\sum_{a\in\caA} \int_{-\infty}^{\infty}dt\ g(t)\Gamma_a(t)^\dag \Gamma_a(t),
\end{equation}
where $g(t)$ is a non-negative integrable function. Clearly, $\rmM^{\rm int}_\caH \ge 0$ and $|\rb^{1/2}\>$ is its ground state with energy 0.

Another family of SoS parent Hamiltonian is characterized by a positive definite matrix. Consider $G \succ 0$ with $G=R^\dag R$ for
an invertible matrix $R$. Define
\begin{equation}
    \rmM_\caH(G) \coloneqq
    \frac{1}{2}
    \sum_{a,b}
    G_{ab}\Gamma_a^\dag\Gamma_b,\qquad \widetilde{\Gamma}_a
    \coloneqq \sum_b R_{ab}\Gamma_b
\end{equation}
Then,
\begin{equation}
    \rmM_\caH(G) =
    \frac{1}{2}
    \sum_{a,b}
    (R^\dag R)_{ab}\Gamma_a^\dag\Gamma_b =
    \frac{1}{2}
    \sum_{a}\widetilde{\Gamma}_a^\dag\widetilde{\Gamma}_a.
\end{equation}
Note that $\rmM_\caH = \rmM_\caH(I)$. As shown in the following lemma, the spectrum of $\rmM_\caH(G)$ is influenced by $G$.

\begin{lemma}
\label{lemma:preconditioning}
Let $\{\Gamma_a\}_{a\in\caA}$ be a finite family of modular
annihilators with $\Gamma_a|\rb^{1/2}\> = 0$, and let
$G\in\bbC^{|\caA|\times|\caA|}$ be Hermitian and positive definite.
Then $\rmM_\caH(G) \succeq 0$ and $\rmM_\caH(G)|\rb^{1/2}\> = 0$.
Moreover,
\begin{equation}
    \lambda_{\min}(G)\,\rmM_\caH(I)
    \preceq
    \rmM_\caH(G)
    \preceq
    \lambda_{\max}(G)\,\rmM_\caH(I).
\end{equation}
Consequently, if $\rmM_\caH(I)$ has a unique ground state, then
\begin{equation}
    \lambda_{\min}(G)\,
    \mathrm{Gap}\!\left[\rmM_\caH(I)\right]
    \le
    \mathrm{Gap}\!\left[\rmM_\caH(G)\right]
    \le
    \lambda_{\max}(G)\,
    \mathrm{Gap}\!\left[\rmM_\caH(I)\right].
\end{equation}
\end{lemma}

\begin{proof}
    See \aref{app:lem1}.
\end{proof}

The target ground state of $\rmM_\caH(G)$ depends only on the common kernel of the modular annihilators, whereas the excited-state spectrum depends on the positive quadratic form used to combine them. This freedom defines a parent-Hamiltonian preconditioning problem: choose the generator-space metric to optimize spectral properties subject to normalization or implementability constraints. We will use free-fermion model as a concrete instance in the following sections.

\begin{remarks}
    The same construction applies to arbitrary state of the form $f^2(H) \otimes I|\Omega\>$, where $f(x)$ is positive and invertible on $\operatorname{spec}(H)$. All we need is to change the modular annihilator to
\begin{equation}
    \Gamma_a = f(H)J_a f(H)^{-1} \otimes I - I \otimes \left[f(H)^{-1}J_a f(H)\right]^\top.
\end{equation}
Then the ground state of $\sum_a \Gamma_a^\dag \Gamma_a$ will be proportional to $f^2(H) \otimes I |\Omega\>$. After tracing out the ancilla system, the reduced state becomes 
\begin{equation}
    \frac{1}{\tr(f^4(H))}\sum_{n=0}^{d-1}f^4(E_n)|\psi_n\>\<\psi_n|
\end{equation}
Hence, this procedure can be viewed as an analogue of quantum eigenvalue transformation protocol \cite{gilyen2019quantum}.
This is an algebraic extension of the modular-annihilator construction.
It identifies a parent Hamiltonian for the filtered purification and does
not provide an efficient state-preparation algorithm.
Implementing the filters, block-encoding the resulting annihilators, and
bounding the relevant gap and initial overlap require additional
assumptions.  For the Gaussian choice
\begin{equation}
f(H)=\exp[-(H-E)^2/(4\sigma^2)],
\end{equation}
the reduced target state is
\begin{equation}
\rho_{E,\sigma}\propto
\sum_n e^{-(E_n-E)^2/\sigma^2}|\psi_n\rangle\langle\psi_n|.
\end{equation}

\end{remarks}

\subsection{From parent Hamiltonian to Lindbladian}
\label{sec:map}

We further demonstrate that under additional conditions the constructed parent Hamiltonian can be mapped back to a Lindbladian. Suppose
\begin{equation}
    \rmM_\caH = \sum_\alpha \eta_\alpha A_\alpha \otimes B_\alpha^\top.
\end{equation}
The original superoperator is
\begin{equation}
    \caH(\cdot) = \sum_\alpha \eta_\alpha A_\alpha (\cdot) B_\alpha.
\end{equation}
By virtue of \eref{eq:superH}, the original Lindbladian should be
\begin{gather}
    \caL(\cdot) = -\sum_\alpha \eta_\alpha \rb^{\frac{1}{4}}A_\alpha\rb^{-\frac{1}{4}} (\cdot) \rb^{-\frac{1}{4}} B_\alpha \rb^{\frac{1}{4}} = -\sum_\alpha \eta_\alpha \Delta^{-1}(A_\alpha) (\cdot) \Delta(B_\alpha).
\end{gather}
In the SoS form, for Hermitian generators $J_a = J_a^\dag$,  after mapping $\rmM_\caH $ to the Lindbladian, we obtain 
\begin{equation}\label{eq:sos}
    \caL_{\mathrm{SoS}}(\cdot) \coloneqq \frac{1}{2}\sum_{a\in\caA}\left[J_a (\cdot) \Delta^2(J_a) + \Delta^{-2}(J_a) (\cdot) J_a\right] - \Delta^{-1}(K) (\cdot) - (\cdot) \Delta(K),
\end{equation}
where
\begin{equation}
    K \coloneqq \frac{1}{2}\sum_{a\in\caA}\Delta(J_a)\Delta^{-1}(J_a).
\end{equation}
To ensure that $\caL_{\rm SoS}$ is a Lindbladian, the generator of a quantum Markov semigroup, we need extra conditions.
\begin{theorem}\label{thm:lindbladian}
    Suppose $\mathrm{Alg}\left(\{J_a\}_{a\in\caA}\right) = \caB((\bbC^2)^N)$ and each $J_a$ is Hermitian. Provided that
    \begin{equation} \label{eq:lindblad_condition}
        \sum_{a\in\caA} J_a \otimes \left(\Delta^2(J_a)\right)^\top = \sum_{a\in\caA}\Delta^{-2}(J_a) \otimes J_a^\top. 
    \end{equation}
    Then $\caL_{\mathrm{SoS}}$ is a primitive Lindbladian.
\end{theorem}
\begin{proof}
    See \aref{app:prooflindblad}.
\end{proof}

The condition \eref{eq:lindblad_condition} can also be written as
\begin{equation}
    \left[\sum_{a\in\caA}J_a \otimes J_a^\top, \rb^{\frac{1}{2}}\otimes \rb^{-\frac{1}{2},\top}\right] = 0.
\end{equation}
The condition is used to prove that $\caL_{\rm SoS}$ satisfies conditional complete positivity.  Another way to prove this is through the Bohr frequency decomposition. We can show that $\caL_{\rm SoS}$ is a Davies generator \cite{Davies1974Markovian}.
\begin{corollary} \label{coro:gkls} Suppose $\mathrm{Alg}\left(\{J_a\}_{a\in\caA}\right) = \caB((\bbC^2)^N)$ and each $J_a$ is Hermitian.  If $\caL_{\rm SoS}$ satisfies \eref{eq:lindblad_condition}, then
\begin{equation}
\caL_{\rm SoS}(\cdot) =\sum_{a\in\caA}\sum_{\nu\in\caF(H)} e^{-\beta\nu/2}\left[ J_{a,\nu}(\cdot)J_{a,\nu}^\dag -\frac{1}{2}\bigl\{J_{a,\nu}^\dag J_{a,\nu},\cdot\bigr\} \right]. 
\end{equation} 
\end{corollary} 
\begin{proof}
    See \aref{app:proof_coro}.
\end{proof}

Next, we give a counterexample to show that the condition \eref{eq:lindblad_condition} is not necessary. We give an exact two-level counterexample. Set $\beta=\log 2$ and $H = -Z$, so that
\begin{equation}
\rho_\beta=\begin{pmatrix}4/5&0\\0&1/5\end{pmatrix},
\end{equation}
Let $J_1=X+Z,J_2 = Y,J_3 = Z$. Then
\begin{equation}
    \sum_{a\in\caA}J_a \otimes J_a^\top = X\otimes X^\top + Y\otimes Y^\top + 2Z \otimes Z^\top + X\otimes Z^\top + Z \otimes X^\top.
\end{equation}
A direct computation shows that the summation does not commute with $\rb^{1/2} \otimes \rb^{-1/2,\top}$. In \aref{app:verification}, we prove that $\caL_{\rm SoS}$ is still a Lindbladian.

\subsection{Example: replacement Lindbladian}

A natural question is how to choose the generators $\{J_a\}$ to maximize the relative spectral gap
\begin{equation}
    r(\{J_a\}) \coloneqq
    \frac{\mathrm{Gap}(\rmM_\caH)}{\|\rmM_\caH\|},
\end{equation}
which is directly related to the simulation cost of quantum Gibbs sampling. In this section, we demonstrate that the answer is trivial without further constraint on $\{J_a\}$.

Let
\begin{equation}
    Z_\beta\coloneqq \Tr(e^{-\beta H}),\qquad \Pi_\beta \coloneqq I - |\rb^{1/2}\>\<\rb^{1/2}|.
\end{equation}
Since $\rmM_\caH\succeq0$ and $\rmM_\caH|\rb^{1/2}\>=0$, one trivially has
$r(\{J_a\})\le1$, with equality iff $\rmM_\caH\propto\Pi_\beta$. The the
next theorem is that the modular-annihilator construction attains this
bound exactly, for every $H$ and every $\beta$.

\begin{theorem}
\label{thm:optimal_generators}
Let $\{W_a\}_{a=1}^{d^2}$ be any Hermitian basis of
$\caB((\bbC^2)^{\otimes N})$ orthonormal in the Hilbert--Schmidt inner
product, $\Tr(W_a W_b)=\delta_{ab}$, and set
\begin{equation}\label{eq:optimal_J}
    J_a = e^{-\beta H/4}W_ae^{-\beta H/4}.
\end{equation}
Then each $J_a$ is Hermitian, $\mathrm{Alg}(\{J_a\})=\caB((\bbC^2)^{\otimes N})$, and
\begin{equation}\label{eq:optimal_parent}
    \rmM_\caH
    =\frac{1}{2}\sum_{a=1}^{d^2}\Gamma_a^\dag\Gamma_a
    = Z_\beta\,\Pi_\beta.
\end{equation}
Consequently $\operatorname{spec}(\rmM_\caH)=\{0,Z_\beta,\ldots,Z_\beta\}$,
\begin{equation}
    \mathrm{Gap}(\rmM_\caH)=\|\rmM_\caH\|=Z_\beta,
    \qquad
    r(\{J_a\})=1.
\end{equation}
\end{theorem}

\begin{proof}
    See \aref{app:replace}.
\end{proof}

The construction of \thref{thm:optimal_generators} also satisfies the
hypothesis of \thref{thm:lindbladian}, and the resulting dissipative
generator identifies exactly why it is not an algorithmically useful
optimum.

\begin{proposition}\label{prop:replacer}
For the generators $\{J_a = e^{-\beta H/4}W_a e^{-\beta H/4}\}_{a=1}^{d^2}$, condition
\eref{eq:lindblad_condition} holds, and
\begin{equation}
    \caL_{\rm SoS}(X)=Z_\beta\left(\Tr(X)\rb - X\right).
\end{equation}
\end{proposition}

\begin{proof}
    See \aref{app:replace}.
\end{proof}

\Pref{prop:replacer} shows that maximal relative energy gap is attained
precisely by the generator family whose associated dynamics is the ideal
sampler. The obstruction to using \thref{thm:optimal_generators} is that $J_a = e^{-\beta H/4}W_a e^{-\beta H/4}$ does not have a closed-form expression in general, and the full set contains exponentially-many terms. Therefore, to make the construction meaningful, we must add constraint on $\{J_a\}$.

\section{Exactly solvable modular transformation}
\label{sec:free_fermion}
Our objective is to construct the parent Hamiltonian without evaluating continuous time integrals or explicitly summing over all Bohr frequencies. This can be achieved whenever one can identify a family of generators $\{J_a:a\in\mathcal A\}$ satisfying three conditions:
\begin{enumerate}
\item \emph{(irreducibility)} the unital algebra generated by $\{J_a\}_{a\in\mathcal A}$ is irreducible;
\item \emph{(efficiency)} the number of generators $|\mathcal A|$ grows at most polylogarithmically with the Hilbert-space dimension;
\item \emph{(tractability)} the modularly dressed operators $\Delta^{\pm1}(J_a)$ can be computed exactly or approximated efficiently.
\end{enumerate}

\subsection{Free fermions}

Free-fermion Hamiltonians provide a natural setting in which all three conditions can be satisfied simultaneously. Consider
$H=\sum_{a,b=1}^{2N}h_{ab}\omega_a\omega_b$,
where $\{\omega_a\}_{a=1}^{2N}$ are Majorana operators and $h$ is the single-particle Hamiltonian matrix in the Majorana basis. Because the commutators $[H,\omega_a]$ close within the linear span of the Majorana operators, imaginary-time evolution maps every Majorana operator to another linear combination of Majorana operators. Consequently, the modular transformation acts linearly on the $2N$-dimensional Majorana space. If each $J_a$ is chosen as a linear combination of Majorana operators, the dressed operators $\Delta^{\pm1}(J_a)$ can therefore be evaluated exactly, establishing tractability. Moreover, only $2N$ generators are required. Since the Hilbert-space dimension is $2^N$, this number is logarithmic in the dimension, establishing efficiency.

In the following paragraphs, we choose
\begin{equation}
    J_a=\sum_{b=1}^{2N}S_{ab}\omega_b,
\end{equation}
where $S$ is an invertible real symmetric matrix satisfying
$[S,h]=0$. Invertibility implies that $\{J_a\}$ spans the same
linear space as the Majorana operators and therefore generates the
full Majorana algebra. 

For the discussion below, we
restrict to the scalar-functional subclass $S=f(h)$, where $f$ satisfies
\begin{equation}
    f: \mathrm{spec}(h) \to \bbR,\qquad f(x = f(-x),\qquad f(x) \neq 0.
\end{equation}
Within this subclass,
the parent Hamiltonian decomposes into independent single-particle
modes, allowing its spectral gap and operator norm to be evaluated
in closed form. We optimize over the admissible family and identify the matrix $S$ minimizing $\|\rmM_\caH\|/\mathrm{Gap}(\rmM_\caH)$,
which provides a natural spectral conditioning measure relevant to SoS-based ground-state preparation and to the relaxation from the maximally mixed state of the associated Lindbladian.

Write the modular annihilator as
\begin{equation}
    \Gamma_a = \sum_{b=1}^{2N} \left(S e^{\beta h}\right)_{ab} \omega_b\otimes I - \sum_{b=1}^{2N} \left(S e^{-\beta h}\right)_{ab} I\otimes \omega_b^\top.
\end{equation}
Its parent Hamiltonian $\rmM_\caH = \sum_{a}\Gamma_a^\dag \Gamma_a/2$ has expansion
\begin{equation}
    \rmM_{\caH} = \frac{1}{2}\sum_{a,b=1}^{2N}\left(S^2 e^{2\beta h}\right)_{ab} (\omega_a \omega_b \otimes I + I \otimes \omega_b^\top \omega_a^\top) - \sum_{a,b=1}^{2N}\left(S^2\right)_{ab}\omega_a \otimes \omega_b^\top.
\end{equation}
By virtue of \thref{thm:lindbladian}, the associated Lindbladian is
\begin{equation}
    \caL_{\mathrm{SoS}}(\cdot) = \sum_{a,b=1}^{2N}\left(S^2 e^{-2\beta h}\right)_{ab} \left(\omega_a (\cdot) \omega_b - \frac{1}{2}\{\omega_b \omega_a, \cdot\}\right).
\end{equation}
The Lindbladian is in the Kossakowski form, and $S^2 e^{-2\beta h}$ is also called the Kossakowski matrix \cite{Gorini1976Completely}. In particular, the DLL parent Hamiltonian belongs to this type.
\begin{theorem}\label{thm:freefermion}
    The DLL parent Hamiltonian is SoS with $S = e^{-2\beta^2 h^2}$.
\end{theorem}

\begin{proof}
    See \aref{app:proof1}.
\end{proof}

Moreover, we observe that this parent Hamiltonian is a sum of projectors as well.

\begin{proposition}\label{prop:projector}
    Let
    \begin{equation}
        R_{+,a} \coloneqq \frac{1}{2}\{\Delta^{-1}(J_a) + [\Delta^{-1}(J_a)]^\dag\},\qquad R_{-,a} \coloneqq \frac{1}{2\rmi}\{\Delta^{-1}(J_a) - [\Delta^{-1}(J_a)]^\dag\}.
    \end{equation}
    Then both $R_{\pm,a}^2$ are proportional to $I$, and
    \begin{equation}
        \left(\Gamma_a^\dag \Gamma_a\right)^2 = 4 (R_{+,a}^2 + R_{-,a}^2)\left(\Gamma_a^\dag \Gamma_a\right).
    \end{equation}
\end{proposition}
\begin{proof}
    See \aref{app:AGSP}.
\end{proof}

In the remaining part of this section, we recover some of the analytic results of free-fermion models in Refs. \cite{Smid2025,smid2026rapid} with new notation.
\begin{proposition}\label{prop:cost}
Suppose
   $ H=\sum_{a,b=1}^{2N}h_{ab}\omega_a\omega_b$
is a free-fermion Hamiltonian, where $h$ has spectrum
$\{\pm\lambda_n/2:n=1,\ldots,N\}$. Let $S=f(h)$, where $f$ is a
real-valued function satisfying
    $f(\lambda_n/2)=f(-\lambda_n/2) \neq 0$
for every $n$. Then
\begin{equation}
\begin{aligned}
    \operatorname{Gap}(\rmM_\caH)
    &=
    2\min_n
    \left\{
        f\left(\frac{\lambda_n}{2}\right)^2
        \cosh(\beta\lambda_n)
    \right\},\\
    \|\rmM_\caH\|
    &=
    4\sum_{n=1}^N
    f\left(\frac{\lambda_n}{2}\right)^2
    \cosh(\beta\lambda_n).
\end{aligned}
\end{equation}
\end{proposition}

\begin{proof}
    See \aref{app:proof_of_gap}.
\end{proof}

Finally, we have the following coefficient matrix that maximizes the relative spectral gap. 
\begin{corollary}\label{coro:opt}
Under the assumptions of Proposition~\ref{prop:cost}, consider the
scalar-functional family $S=f(h)$. If
    $\|\rmM_\caH\|=N$,
then
    $\operatorname{Gap}(\rmM_\caH)\le 1/2$.
This upper bound is attained by
\begin{equation}\label{eq:Sopt}
    S=S_{\mathrm{opt}}
    \coloneqq
    \frac{1}{2}\cosh(2\beta h)^{-1/2}.
\end{equation}
\end{corollary}

\begin{proof}
     By \pref{prop:cost}, we have $\|\rmM_\caH\| \ge 2N \mathrm{Gap}(\rmM_\caH)$, hence $\mathrm{Gap}(\rmM_\caH) \le 1/2$. When $f(x) = 1/(2\sqrt{\cosh(2\beta x)})$, it is direct to verify that $f(x) = f(-x)$ for all $x$, $\|\rmM_\caH\| = N$ and $\mathrm{Gap}(\rmM_\caH) = 1/2$. 
\end{proof}

In conjunction with \thref{thm:lindbladian}, the Lindbladian that uses the coefficient matrix in \eref{eq:Sopt} is
\begin{equation}
    \caL_{\mathrm{SoS}}(\cdot) = \frac{1}{4}\sum_{a,b=1}^{2N}\left[\cosh(2\beta h)^{-1} e^{-2\beta h}\right]_{ab} \left(\omega_a (\cdot) \omega_b - \frac{1}{2}\{\omega_b\omega_a, \cdot\}\right).
\end{equation}
We can also analytically solve the mixing time upper bound of this Lindbladian.
\begin{proposition}\label{prop:mixing_time}
Suppose $H = \sum_{a,b=1}^{2N}h_{ab}\omega_a \omega_b$ is a free-fermion model, where $h$ has spectrum $\{\pm \lambda_n/2 : n=1,\ldots,N\}$. Let $S=f(h)$, where $f(x)$ is a real-valued function satisfying $f(\lambda_n/2) = f(-\lambda_n/2)$ for all $n$. Then for every initial state $\rho_0$, 
\begin{equation}
\label{eq:worst_case_ff_simple}
    \frac12
    \left\|
    e^{t\mathcal L_{\rm SoS}}(\rho_0)-\rho_\beta
    \right\|_1
    \le
    2N e^{-g_\star t},\qquad g_{\star} \coloneqq \min_n 2f(\lambda_n/2)^2\cosh(\beta\lambda_n). 
\end{equation}
Therefore,
\begin{equation}\label{eq:mixing}
    t_{\rm mix}(\epsilon)
    \le
    \frac{1}{g_\star}
    \log\left(\frac{2N}{\epsilon}\right).
\end{equation}
\end{proposition}
\begin{proof}
    See \aref{app:proof_of_mixing}.
\end{proof}
In particular, the choice $f(h) = 1/(2\sqrt{\cosh(2\beta h)})$ has mixing time upper bound $2\log(2N/\epsilon)$, which is independent of the inverse temperature $\beta$. 
    
\section{Approximate modular transformation for Interacting Systems}
\label{sec:krylov}

The modularly dressed operator $\Delta(J_a)$ is tractable whenever $J_a$ belongs to a low-dimensional subspace that is invariant under commutation with $H$, namely,
$[H,J_{a}]\in \operatorname{span}\{J_{a'}:a'\in\caA\}$.
Within such a subspace, the modular transformation reduces to exponentiating the finite-dimensional matrix $c$. For generic interacting many-body systems, however, no closed-form expression for $\Delta(J_a)$ is generally available. The modular-annihilator construction is therefore exact but implicit: an efficient procedure is still required to evaluate the modularly dressed operators entering the parent Hamiltonian. To address this problem, we employ the Krylov--Lanczos recursion method \cite{Hochbruck1997Krylov} to approximate the imaginary-time Heisenberg evolution and thereby obtain efficient approximations to $\Delta(J_a)$.

We first quantify how errors in the modular transformations propagate to the ground state of the resulting approximate parent Hamiltonian. The following proposition bounds the ground-state error in terms of the approximation error of the modularly dressed generators.

    \begin{proposition}\label{prop:spectral_gap}
        Suppose $\mathrm{Alg}(\{J_a\}_{a\in\caA})$ is an unital irreducible algebra on $(\bbC^2)^{\otimes N}$ and $\rmM_\caH$ is an SoS parent Hamiltonian of the form \eref{eq:modular_annihilator} generated from $\{J_a\}_{a\in\caA}$ and each $J_a$ is Hermitian. It has ground state $|\Psi\>$ and spectral gap $\mathrm{Gap}(\rmM_\caH) = g$. Let $\widetilde{\Delta^{\pm 1}}(J_a)$ be the approximation of $\Delta^{\pm 1}(J_a)$ with operator norm error bound $\|\widetilde{\Delta^{\pm 1}}(J_a) - \Delta^{\pm 1}(J_a)\| \le \varepsilon$. Let 
        \begin{equation}
            \widetilde{\rmM}_\caH \coloneqq \frac{1}{2}\sum_{a\in\caA}\left(\widetilde{\Delta^{-1}}(J_a) \otimes I - I\otimes \widetilde{\Delta}(J_a)^\top\right)^\dag \left(\widetilde{\Delta^{-1}}(J_a) \otimes I - I\otimes \widetilde{\Delta}(J_a)^\top\right).
        \end{equation}
        If 
        \begin{equation}
            \sum_{a\in \caA} (\|\Gamma_a\|\varepsilon + \varepsilon^2) \le \frac{g}{8},
        \end{equation}        
        then $\widetilde{\rmM}_\caH$ has a unique ground state $|\widetilde{\Psi}\>$ with spectral gap $\mathrm{Gap}(\widetilde{\rmM}_\caH) \ge g/2$, and
        \begin{equation}
            1 - |\<\Psi|\widetilde{\Psi}\>|^2 \le \frac{4|\caA|\varepsilon^2}{g}.
        \end{equation}
    \end{proposition}

    \begin{proof}
        See \aref{app:proof_prop}.
    \end{proof} 

\subsection{Krylov--Lanczos approximation}

Consider the space of linear operators $\caB((\bbC^2)^{\otimes N})$. Define inner product $\<x,y\> = \Tr(x^\dag y)$. Let $V_1\in \caB((\bbC^2)^{\otimes N})$ be a normalized operator such that $\<V_1,V_1\> = \|V_1\|_F^2 = 1$. Let $\mathrm{ad}_H(\cdot)$ be the Hermitian superoperator. Hence, the Krylov subspace is 
\begin{equation}
    \caK_m\left(\mathrm{ad}_H(\cdot),V_1\right) \coloneqq \mathrm{span}\left\{V_1, \mathrm{ad}_H(V_1),\ldots,\mathrm{ad}_H^{m-1}(V_1)\right\}.
\end{equation}
the Lanczos recursion is accomplished by
\begin{enumerate}
    \item $ W_k = \mathrm{ad}_H(V_k) - b_{k-1}V_{k-1},$
    \item $a_k = \Tr(V_k^\dag W_k),$ 
    \item $W_k \leftarrow W_k-a_kV_k,$
    \item $b_k = \|W_k\|_F,$
    \item $V_{k+1} = W_k/b_k,$
\end{enumerate}
with initial condition $b_0 = 0$. We use the following notation to denote the result of Lanczos recursion:
\begin{equation}
    e^{x H} V_1 e^{-xH} = e^{ x\cdot \mathrm{ad}_H}(V_1) \approx  L_V\left(e^{x T_m }t_1\right),\qquad L_V(t_j) \coloneqq V_j \quad \forall j.
\end{equation}
where $T_m$ is the matrix representation of $\mathrm{ad}_H(\cdot)$ obtained from $m$ steps of Lanczos recursion, and $\{t_j\}_{j=1}^m$ is the vector basis of $T_m$. 

 Let $\caK(m,\operatorname{ad}_H, x,V) \coloneq \sqrt{d}L_V(e^{x T_m}t_1)$ be the approximated modular transformation obtained by the $m$ step Lanczos recursion.  Up to this step, we derive the algorithmic error by Krylov subspace method.
\begin{theorem}[Worst-case error bound]\label{thm:worst}
    Suppose $H$ is an $N$-qubit Hamiltonian. Let $V$ be an $1$-local Pauli operator. Then in order to have error bound
\begin{equation}
    \|\Delta(V) - \caK(m,\operatorname{ad}_H,\beta/4,V) \|_F \le \sqrt{2^N}\varepsilon_F,
\end{equation}
it suffices to take
\begin{equation}
    m = O\left(\|H\|\beta + \log(\varepsilon_F^{-1})\right).
\end{equation}
\end{theorem}
\begin{proof}
    Because $V$ is a Pauli operator, we have $\|V\|_F = \sqrt{d}$. By virtue of \lref{lem:polyapprox}, in order to have $\|e^{xH}V_1 e^{-xH} - L_V( e^{xT_m}t_1)\|_F \le \varepsilon$, it suffices to take
    \begin{equation}
        m = O\left(x\|\mathrm{ad}_H(\cdot)\|_{F \to F} + \log(\varepsilon^{-1})\right).
    \end{equation}
    Note that the induced metric satisfies 
    \begin{equation}
        \|\mathrm{ad}_H(\cdot)\|_{F \to F} = \max_{X\neq 0}\frac{\|[H,X]\|_F}{\|X\|_F} \le 2\|H\|.
    \end{equation}
    Then the theorem is proved by setting $V_1 = V/\sqrt{d}, x = \beta/4,$ and $\varepsilon = \varepsilon_F$. 
\end{proof}
Note that the bound in \thref{thm:worst} is in terms of the Frobenius 2-norm, but we need the operator norm upper bound in \pref{prop:spectral_gap}. In the next section, we will use the locality argument to provide an operator norm upper bound that does not depend on the full system dimension $d$.

\subsection{Locality-improved bound}

The worst-case estimate obtained in the previous subsection is generally too pessimistic for local quantum many-body systems. When the initial operator is local, one can combine imaginary-time quasi-locality with the Krylov approximation to obtain substantially improved bounds.

Let $V$ be a single-site Pauli operator supported on lattice site \(i\), and consider a one-dimensional Hamiltonian with interaction range $r$ and local interaction strength bounded by $J$:
\begin{equation} \label{eq:finiterange}
H = \sum_j h_j,
\qquad \mathrm{diam}\left(\operatorname{supp}(h_j)\right) \le r,\qquad 
J\coloneqq\sup_{\ell\in\Lambda}
\sum_{j:\ell\in \operatorname{supp}(h_j)}\|h_j\|,
\end{equation}
where $J$ is independent of the total system size. We assume $\operatorname{supp}(h_j) \neq \operatorname{supp}(h_k)$ for all $j \neq k$. 
Let \(\Lambda_R(i)\) denote the interval of radius \(R\) centered at site \(i\), and define the truncated Hamiltonian 
\begin{equation}
H_{R,i} \coloneqq \sum_{\operatorname{supp}(h_j)\subseteq \Lambda_R(i)} h_j .
\end{equation} 
The Araki-type locality estimate
of Ref.~\cite{perez2023locality} tells how large $R$ should be to keep
\begin{equation}
\left\|
e^{xH}Ve^{-xH}-e^{xH_{R,i}}Ve^{-xH_{R,i}}
\right\| \le \varepsilon_{\rm LR}.
\end{equation}
Hence, we can directly apply the Lanczos approximation to the finite-volume evolution generated by \(H_{R,i}\) instead of $H$. The core idea is to replace $H$ with $H_{R,i}$ and limit attention to the subspace $\caH_{R} \coloneq \bigotimes_{j\in \Lambda_R(i)}\bbC^2$, which contains only $2R+1$ qubits. We can define an effective Frobenius norm $\|X\|_{R,F}$ on this subspace. Since $\|V\|_{R,F} = \sqrt{2^{2R+1}}$ for any Pauli operator $V$ on $\caH_R$, the Krylov approximation of modularly dressed operator becomes
\begin{equation}
    \caK(m,\operatorname{ad}_{H_{R,i}},\beta/4,V) \coloneq \sqrt{2^{2R+1}}L_V(e^{\beta T_m/4}t_1).
\end{equation}
Hence, the dependence on the total system size $N$ is removed. In conjunction with \lref{lem:polyapprox} in \aref{app:lanczos}, we obtain the following result.

\begin{theorem}[Locality-improved error bound]\label{thm:local}
Suppose $H$ is an $N$-qubit Hamiltonian of the form \eref{eq:finiterange}, and $V$ is a Pauli operator supported on the $i$th site, and $\beta < 1/(Jr+J)$. Then for error tolerance \(0 < \varepsilon < e^{-1}\), there exists a truncated Hamiltonian \(H_{R,i}\) supported on a neighborhood of the $i$th site of radius
\begin{equation}
    R = O(\beta Jr^3 + r\log(\varepsilon^{-1})),
\end{equation}
such that the \(m\) step Lanczos approximation satisfies
\begin{equation}
    \|\Delta(V)-\caK(m,\operatorname{ad}_{H_{R,i}},\beta/4,V)\| \le \varepsilon.
\end{equation}
provided that
\begin{equation}
m=O\left(R+ \log(\varepsilon^{-1})\right) = O(\beta J r^3 + r\log(\varepsilon^{-1})).
\end{equation}
\end{theorem}

\begin{proof}
    See \aref{app:local}. 
\end{proof}
Note that the implicit constants (except $\varepsilon$) depend only on the locality structure of the Hamiltonian and are independent of the system size $N$. At a cost, the approximation is only valid in high temperature regime $\beta < 1/(Jr+J)$. 

\subsection{Krylov-approximated dissipative generator}

The Krylov approximation of the modular transformation can also be
propagated to the corresponding dissipative generator. We emphasize,
however, that a finite Krylov truncation does not in general preserve
conditional complete positivity. Therefore, without an additional
positivity condition, the resulting object should be regarded as an
approximate dissipative generator rather than an exact Lindbladian.

For each $J_a$, there exists a sequence of operators $\{J_{a}^{(i)}, i=0,1,\ldots,m\}$, such that $J_a^{(0)}  =J_a$, and
\begin{equation}
    \caK(m, \mathrm{ad}_H, \beta/4, J_a) = \sum_{i=0}^m  d_i J_a^{(i)}. 
\end{equation}
Therefore, consider the operators spanned by $\{J_a^{(i)}:a\in\caA, i = 0,\ldots,m\}$, denote its basis as $\{L_\alpha\}$. Then the SoS superoperator can be written in the form of 
\begin{equation}
    \caL_{\rm SoS}(\cdot) = \sum_{\alpha,\beta} \Xi_{\alpha,\beta}\left[L_\alpha (\cdot)L_\beta^\dag - \frac{1}{2}\{L_\beta^\dag L_\alpha,\cdot\}\right].
\end{equation}
This superoperator is a Lindbaldian when $\Xi \ge 0$. Hence, we numerically find the optimal approximation by
\begin{equation}
    \Xi_* \coloneqq \arg\min_{X: X \ge 0}\|X - \Xi\|_F.
\end{equation}
Then we can approximate the Lindbladian with matrix $\Xi_*$. However, without further assumptions, there is no quantitative result on the error bound of this approximation. We leave it for future work.

\section{Numerical Validation}
\label{sec:numerical}

In \sref{sec:num_ff}, we numerically verify the observations in Propositions~\ref{prop:cost} and \ref{prop:mixing_time} with a free-fermion model with $N = 100$ modes. Numerical results illustrate that for a random free-fermion model, the coefficient matrix $S_{\rm opt}$ has a smaller reduced QSVT cost proxy and faster relaxation from the maximally mixed state of Lindbladian evolution. In \sref{sec:num_krylov}, we apply the Krylov-Lanczos approximation to an interacting fermionic model, and illustrate finite-size accuracy in weak interaction or high temperature regime.

\subsection{Free-fermion models with different coefficient matrices}
\label{sec:num_ff}

Recall that with a fixed coefficient matrix $S$, the parent Hamiltonian is constructed by
\begin{equation}
    J_a = \sum_b S_{ab}\omega_b, \qquad \Gamma_a = \Delta^{-1}(J_a)\otimes I - I \otimes \Delta(J_a)^\top,\qquad \rmM_\caH = \frac{1}{2}\sum_a \Gamma_a^\dag \Gamma_a.
\end{equation}
The Lindbladian $\caL$ is the unvectorization of the matrix form
\begin{equation}
    \rmM_\caL = -\left(\rb^{\frac{1}{4}}\otimes \rb^{\frac{1}{4},\top}\right)\rmM_\caH\left(\rb^{-\frac{1}{4}}\otimes \rb^{-\frac{1}{4},\top}\right).
\end{equation}

In this subsection, we compare the performance of three different coefficient matrices:
\begin{equation}
    S_{\mathrm{Gaussian}} = e^{-2\beta^2 h^2},\qquad S_{\mathrm{id}} = I,\qquad S_{\mathrm{opt}} = \frac{1}{2\sqrt{\cosh(2\beta h)}}.
\end{equation}
For each $S$, there are two ways to obtain the Gibbs state: the first is to prepare the ground state of $\rmM_\caH$ using quantum singular value transformation (QSVT) \cite{gilyen2019quantum}, the second is to simulate the Lindbladian evolution. We discuss how to quantify their performance separately. In particular, \fref{fig:costcompareN100}a is about the QSVT cost of ground state preparation, and \fref{fig:costcompareN100}b is about the relaxation from the maximally mixed state of the Lindbladian.

Thanks to the SoS factorization \cite{king2026quantum}, the cost of state preparation by QSVT can be made more efficient than the vanilla method \cite{Lin2020near}. Observe that the parent Hamiltonian can be written as
\begin{equation}
\rmM_{\caH}
=
\frac{1}{2}\sum_{a\in\caA}\Gamma_a^\dagger\Gamma_a
=
\rmB_{\caH}^\dagger\rmB_{\caH}
\ge 0,
\qquad
\rmB_{\caH}
\coloneqq
\frac{1}{\sqrt{2}}\sum_{a\in\caA}\ket{a}\otimes\Gamma_a .
\end{equation}
Since $\Gamma_a|\rb^{1/2}\>=0$ for all $a$, the purified Gibbs state lies in $\ker(\rmM_{\caH})=\ker(\rmB_{\caH})$. Under the irreducibility assumption, this kernel is one-dimensional, so $|\rb^{1/2}\>$ is the unique ground state.

Suppose an $\alpha_B$-normalized block encoding of $\rmB_{\caH}$ is available, with $\alpha_B
\ge
\|\rmB_{\caH}\|
=
\sqrt{\|\rmM_{\caH}\|}$.
The normalized singular-value gap separating the kernel of $\rmB_{\caH}$ from the remaining singular subspaces is
$\sqrt{\operatorname{Gap}(\rmM_{\caH})}\alpha_B^{-1}$.
Hence, QSVT can approximate the projector onto the ground state to accuracy $\varepsilon$ using
\begin{equation}
O\left(
\frac{\alpha_B}{
\sqrt{\operatorname{Gap}(\rmM_{\caH})}
}
\log\frac{1}{\gamma\varepsilon}
\right)
\end{equation}
queries to the block encoding of $\rmB_{\caH}$.

If one block-encoding query costs $C_{\rm BE}(\beta)$ and the initial state has ground-state overlap $\gamma = |\<\rb^{1/2}|\psi_{\rm input}\>|$, the total preparation cost is
\begin{equation}
O\left(
\frac{
C_{\rm BE}(\beta)\alpha_B
}{
\gamma\sqrt{\operatorname{Gap}(\rmM_{\caH})}
}
\log\frac{1}{\gamma\varepsilon}
\right).
\end{equation}
In the favorable case $C_{\rm BE}(\beta)=O(\beta)$, $\gamma=\Theta(1)$, and $\alpha_B=\Theta(\sqrt{\|\rmM_{\caH}\|})$, the total cost becomes
\begin{equation}
O\left(
\beta
\sqrt{
\frac{
\|\rmM_{\caH}\|
}{
\operatorname{Gap}(\rmM_{\caH})
}
}
\log\frac{1}{\varepsilon}
\right).
\end{equation}
At fixed precision, we therefore use
\begin{equation}\label{eq:cost}
\operatorname{cost}(\beta)
\coloneqq
\beta
\sqrt{
\frac{
\|\rmM_{\caH}\|
}{
\operatorname{Gap}(\rmM_{\caH})
}
}
\end{equation}
as a reduced cost proxy.

As to the performance of $e^{\caL_{\rm SoS} t}$, we directly compute the upper bound of trace-distance obtained in \eref{eq:trdistance_bound} where the initial state is fixed as $\rho(0) = I/2^N$. Note that the true trace-distance $\|\rho(t) - \rb\|_1/2$ should never exceed 1. However, we do not include this constant upper bound in the illustration.

In \fref{fig:costcompareN100}, we compare the reduced simulation proxy of the first method, and the relaxation from the maximally mixed state of the second method. The model we choose is a random free-fermion model with $N=100$ modes sampled as follows. Let $g$ be $2N \times 2N$ matrix with $g_{ab} \sim \caN(0,1)$. Set $h = \rmi (g - g^\top)$ and normalize it to $\|h\|_1 = 1$. The random free-fermion model is thus $H = \sum_{a,b=1}^{2N}h_{ab}\omega_a \omega_b$. As shown by the numerical tests in \fref{fig:costcompareN100}, the performance of $S_{\mathrm{opt}}$ is better than the other two in both scenarios. In terms of the cost of QSVT state preparation, the function cost$(\beta)$ of $S_{\rm opt}$ increases slower than those of $S_{\rm id}$ and $S_{\rm Gaussian}$. As to the relaxation from the maximally mixed state, the trace-distance upper bound decreases much faster with $\|\rmM_\caH\|t$ when we choose $I/2^N$ as the initial state. 

\begin{figure}[htbp]
    \centering
    \includegraphics[width=0.6\linewidth]{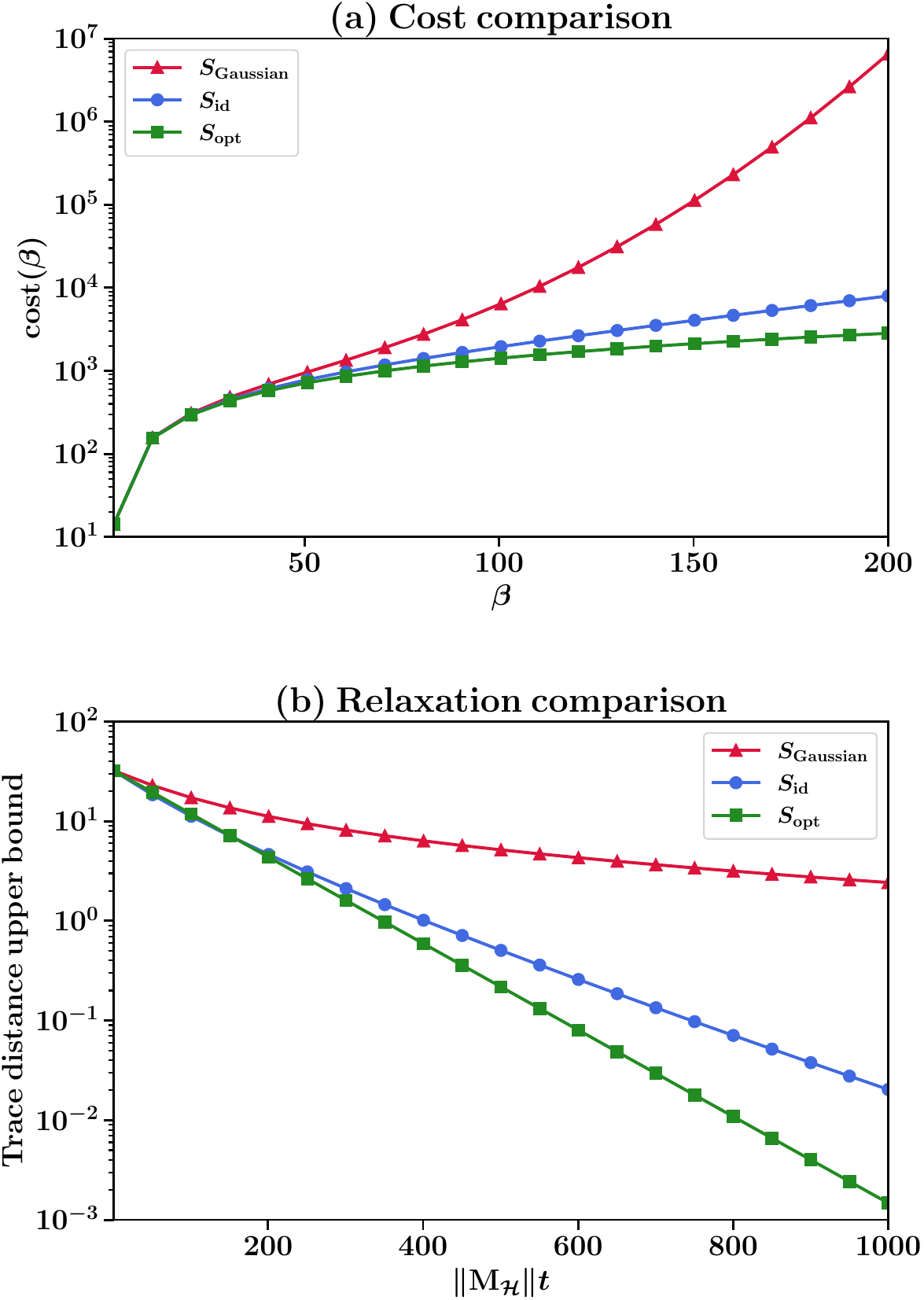}
    \caption{The comparison between the three coefficient matrices, where the free-fermion model is generated by random coefficients and renormalization $\|h\|_1 = 1$. There are $N = 100$ modes in the model. In plot (a), the horizontal axis is the inverse temperature $\beta\in [1,200]$. The vertical axis is cost$(\beta)$ defined in \eref{eq:cost}, which reflects the reduced query-cost proxy using SoS spectral amplification. In plot (b), we compare the relaxation from the maximally mixed state of the corresponding Lindbladian, and fix the inverse temperature to be $\beta = 100$. The horizontal axis is the evolution time multiplied by $\|\rmM_\caH\|$, the vertical axis is the analytical upper bound of the trace distance derived in \eref{eq:mixing}.}
    \label{fig:costcompareN100}
\end{figure}

\subsection{Krylov-Lanczos approximation of interacting fermion models}
\label{sec:num_krylov}

In this subsection, we numerically evaluate the efficiency of the Krylov--Lanczos approximation by examining the operator-norm error of the modularly dressed operators, the fidelity between the exact and approximate ground states, and the spectral gaps of the corresponding parent Hamiltonians.
 The model we use is
\begin{equation}\label{eq:ff_int}
    H = \sum_{a,b=1}^{2N}h_{ab}\omega_a \omega_b + U\sum_{a=1}^{N-1} \left(n_a - \frac{I}{2}\right)\left(n_{a+1} - \frac{I}{2}\right)
\end{equation}
where
\begin{equation}
    c_a \coloneqq \frac{1}{2}(\omega_{2a-1} + \rmi \omega_{2a}),\qquad n_a \coloneqq c_a^\dag c_a.
\end{equation}
We set $N = 4$ and $U = 10^{-2}, 10^{-1}, 1, 10$. The single-particle matrix $h$ is sampled in the same way as in \fref{fig:costcompareN100}. In the Krylov-Lanczos approximation, we set $m = 8,16,24$, and compute the estimation of $\Delta^{\pm 1}(\cdot)$ separately, instead of treating $\Delta^{-1}(\cdot)$ as the Hermitian conjugate of $\Delta(\cdot)$. The set of generators is chosen as $\{J_a\} = \{X_n, Z_n\}_{n=1}^{N}$. We quantify the performance of the approximation from three aspects: the maximal operator norm distance of modularly dressed operators 
\begin{equation}\label{eq:max_eps}
    \epsilon_{\max} \coloneqq \max_a \|\Delta^{\pm 1}(J_a) - \caK(m, \mathrm{ad}_H, \pm\beta/4, J_a)\|,
\end{equation}
the ground state fidelity, and the spectral gap of the parent Hamiltonian. We verified that both the exact and Krylov-approximated parent Hamiltonians have numerically nondegenerate ground states throughout the reported parameter range, with the first excitation gap exceeding the prescribed degeneracy tolerance. Figures~\ref{fig:krylov}, \ref{fig:gap}, and \ref{fig:fidelity} consistently show that the Krylov approximation remains highly accurate at moderate inverse temperature $\beta$, but eventually breaks down as $\beta$ or the coupling strength $U$ increases. Enlarging the Krylov dimension $m$ substantially postpones this breakdown. 

\begin{figure}[htbp]
    \centering
    \includegraphics[width=0.75\linewidth]{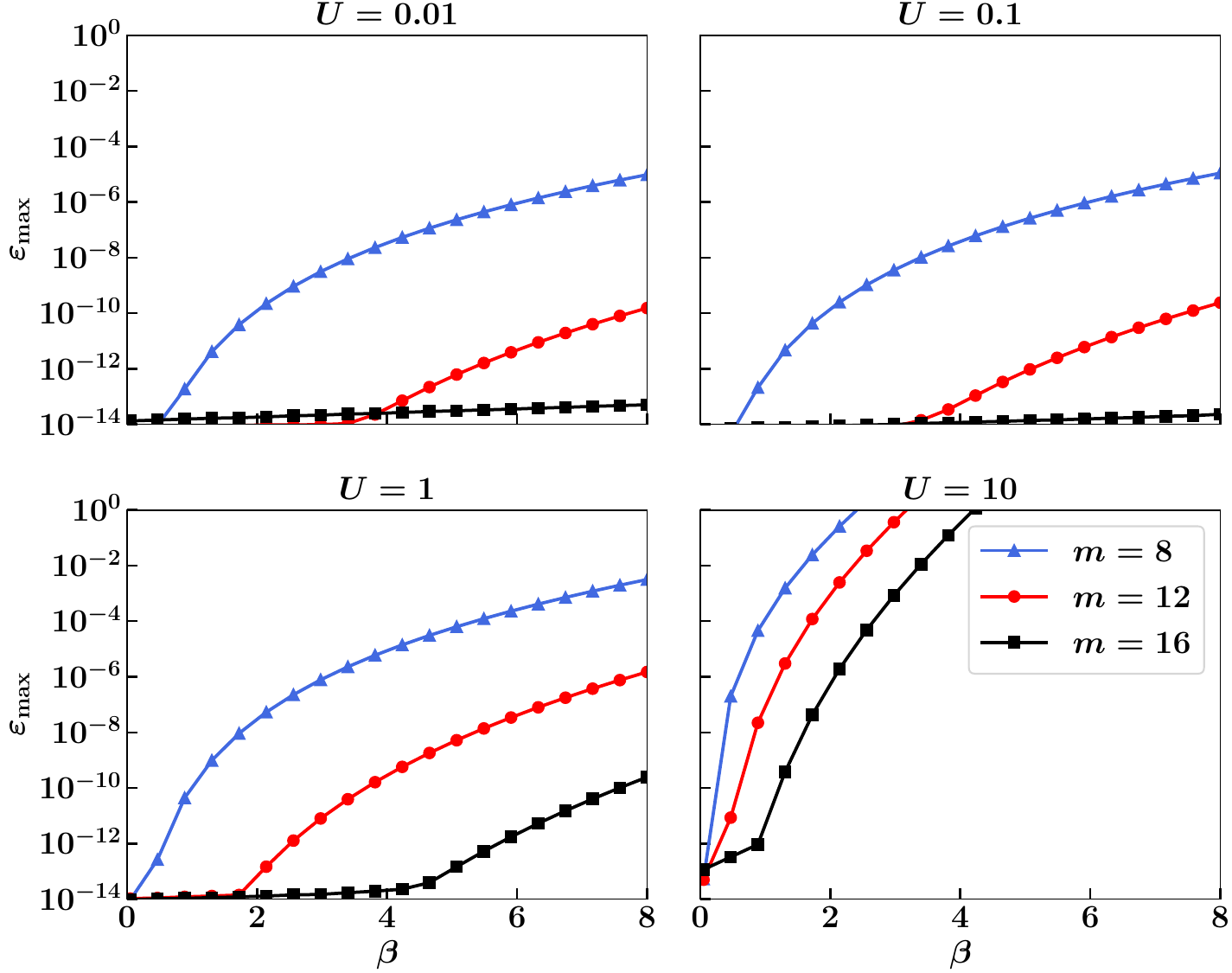}
    \caption{The accuracy of approximate modular transformation of interacting fermionic model in \eref{eq:ff_int} with system size $N = 4$ and interaction strength $U = 10^{-2}, 10^{-1}, 1, 10$. The horizontal axis corresponds to the inverse temperature $\beta$. The vertical axis corresponds to $\epsilon_{\max}$ in \eref{eq:max_eps} for $m = 8,16,24$ separately. The set of generators is 1-local Pauli operators $\{X_n,Z_n\}_{n=1}^N$.}
    \label{fig:krylov}
\end{figure}

\begin{figure}[htbp]
    \centering
    \includegraphics[width=0.75\linewidth]{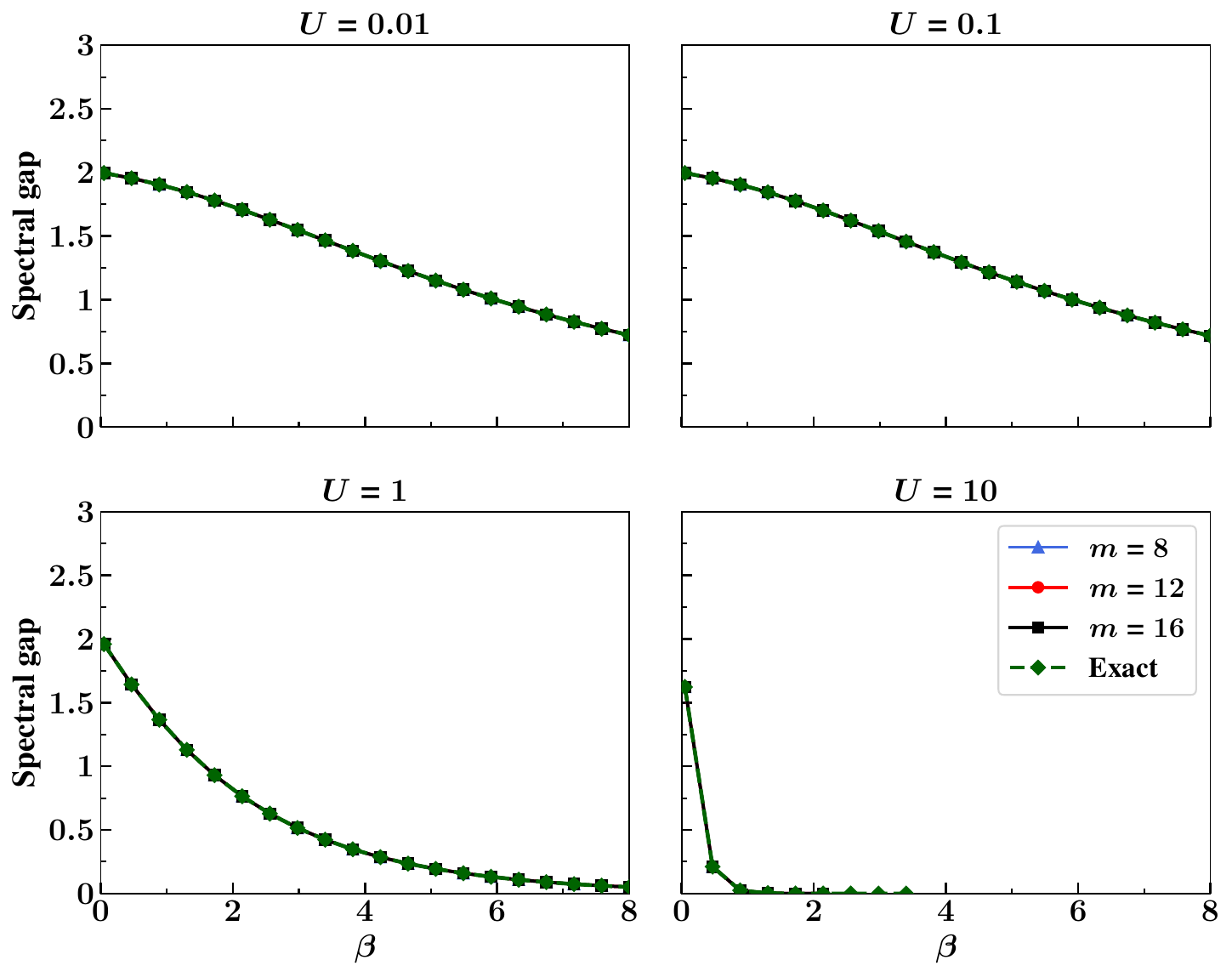}
    \caption{Spectral gap of the exact parent Hamiltonian and Krylov approximated parent Hamiltonian with system size $N = 4$ and interaction strength $U = 10^{-2}, 10^{-1}, 1, 10$. The horizontal axis corresponds to the inverse temperature $\beta$. The vertical axis corresponds to the spectral gap of $\rmM_\caH$ and $\widetilde{\rmM}_\caH$ for $m = 8,16,24$ separately. The set of generators is 1-local Pauli operators $\{X_n,Z_n\}_{n=1}^N$. We cut off the computation when the spectral gap of the Krylov approximated parent Hamiltonian at $m = 24$ is less than or equal to $10^{-8}$, which is considered as having degenerate ground space in the numerical sense. Note that the exact curve is visually indistinguishable from $m=24$ curve over the displayed range.}
    \label{fig:gap}
\end{figure}

\begin{figure}[htbp]
    \centering
    \includegraphics[width=0.75\linewidth]{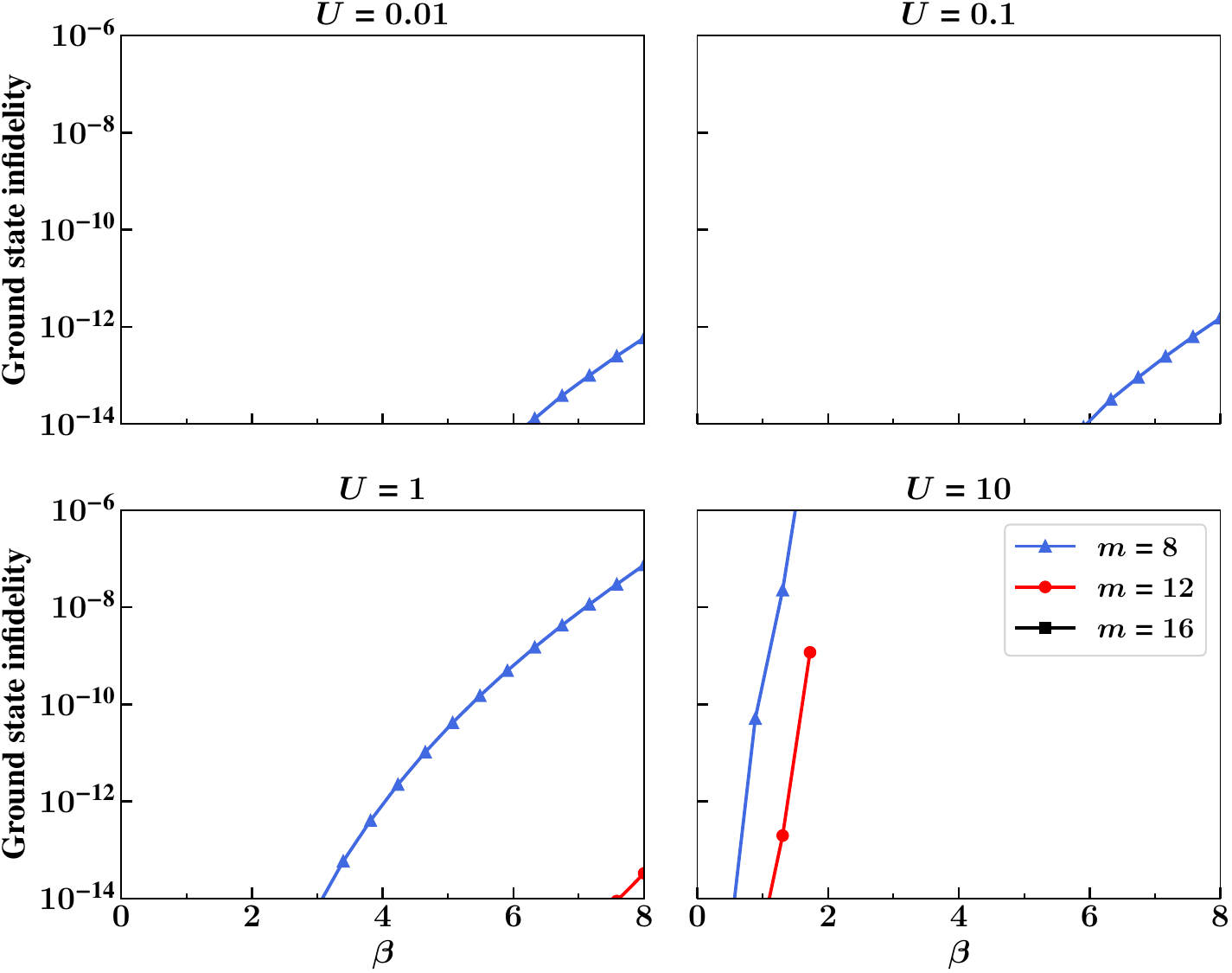}
    \caption{Ground state fidelity between the ground state of the exact parent Hamiltonian and Krylov approximated parent Hamiltonian with system size $N = 4$ and interaction strength $U = 10^{-2}, 10^{-1}, 1, 10$. The horizontal axis corresponds to the inverse temperature $\beta$. The vertical axis corresponds to the fidelity between $|\rb^{1/2}\>$ and the ground state of $\widetilde{\rmM}_\caH$ for $m = 8,16,24$ separately. The set of generators is 1-local Pauli operators $\{X_n,Z_n\}_{n=1}^N$. We cut off the computation when the spectral gap of the Krylov approximated parent Hamiltonian at $m = 24$ is less than or equal to $10^{-8}$, which is considered as having degenerate ground space in the numerical sense.}
    \label{fig:fidelity}
\end{figure}

\section{Conclusion and outlook}
\label{sec:conclusion}

In this work, we developed a unified framework based on modular transformations for constructing sum-of-squares parent Hamiltonians of purified Gibbs states. The resulting factorization identifies the purified Gibbs state as the common zero mode of a discrete family of modular annihilators, thereby making positivity and frustration freeness explicit while avoiding continuous-time integrals and full Bohr-frequency decompositions. Because the same set of generators determines both the parent Hamiltonian and the associated dissipative dynamics, this construction establishes a direct connection between ground-state preparation and Lindbladian thermalization. We further proved that the purified Gibbs state is the unique ground state whenever the generators form a unital irreducible algebra (\thref{thm:algebra}). In addition, positive definite weightings of the modular annihilators provide a natural preconditioning freedom that modifies the excited-state spectrum and spectral gap without changing the target ground state (\lref{lemma:preconditioning}).

For free-fermion systems, we showed that the modular transformation is exactly solvable and that our framework produces an entire family of SoS parent Hamiltonians parameterized by a real symmetric coefficient matrix $S$. In particular, the choice $S_{\mathrm{Gaussian}} \propto e^{-2\beta^2 h^2}$ recovers the canonical DLL parent Hamiltonian. Exploiting this structure, we derived closed-form expressions for the
spectral gap and operator norm for the scalar-functional subclass
$S=f(h)$, and proved that, within this subclass, the spectral gap at
fixed operator norm is maximized by
$S_{\mathrm{opt}}\propto 1/\sqrt{\cosh(2\beta h)}$ (\coref{coro:opt}). Numerical experiments on random free-fermion models supported this observation, in terms of both the QSVT-based ground state preparation cost and the relaxation from the maximally mixed state of the associated Lindbladian evolution.

To address generic interacting systems, where the modularly dressed operators admit no closed form, we developed a Krylov--Lanczos approximation of the imaginary-time evolution. We proved a perturbation lemma bounding the ground state infidelity of the approximate parent Hamiltonian in terms of the modular transformation approximation error and the spectral gap (\pref{prop:spectral_gap}). To keep the dimension-normalized Frobenius error bounded by $\varepsilon$, the required Krylov dimension scales linearly as $O(\|H\|\beta + \log(\varepsilon^{-1}))$ (\thref{thm:worst}); for a one-dimensional Hamiltonian with interaction range $r$ and local interaction strength $J$, combining the Krylov method with Araki-type imaginary-time quasi-locality improves the requirement to $O(\beta J r^3 + r\log(\varepsilon^{-1}))$ provided that $\beta J(r+1) < 1$ (\thref{thm:local}). Numerical simulations of interacting fermion models support the predicted dependence on the inverse temperature, the interaction strength, and the Krylov dimension.

Several directions remain open. First, our locality-improved Krylov bound relies on relatively coarse quasi-locality estimates and still grows polynomially with $\beta$. A more refined understanding of operator growth in the Krylov basis may lead to sharper bounds. Second, the sufficient condition in \thref{thm:lindbladian} has so far been verified only for the replacement Lindbladian and the free-fermion family. Characterizing the generator sets $\{J_a\}$ that satisfy this condition beyond exactly solvable models is therefore an important direction for future work. Third, we observe that all SoS parent Hamiltonians in the free-fermion setting are sum of weighted projectors, which suggests a potential construction of approximate ground-state projectors (AGSP) through the detectability lemma \cite{anshu2016simple,fang2026quantum}. The implications of this structure for the entanglement properties of Gibbs states \cite{kuwahara2021improved,kim2025thermal} deserve further investigation. More broadly, the modular-annihilator framework provides a unified perspective connecting Gibbs state preparation, Lindbladian dynamics, and operator algebras, and opens new avenues for studying finite-temperature quantum many-body systems.

\section{Data Availability}

The code and data used to generate the results reported
in this manuscript are publicly available at Ref.~\cite{github2026}. 

\section{Acknowledgement}

C.Y. acknowledges support from the Shanghai Municipal Science and Technology Major Project (Grant No. 2019SHZDZX01-ZX04) and the institutional research funding of Shanghai University under Project No. N.14-G204-26-603, and J.T. was supported by the JSPS KAKENHI Grant Numbers JP25K17310, JP25H01391, JP25H01388, JP25K24845, and JP25K24848. C.Z. acknowledges support from the Natural Sciences and Engineering Research Council of Canada (NSERC) through Discovery Grant RGPIN-2026-06304; from the Faculté des sciences and the Institut quantique at Université de Sherbrooke; and from the Institut transdisciplinaire d’information quantique (INTRIQ), a strategic cluster funded by the Fonds de recherche du Québec – Nature et technologies. The authors used GPT-5.6 Sol and Fable 5 to assist with language editing and with checking and refining mathematical arguments. All AI-assisted outputs were critically reviewed and, where appropriate, revised by the authors, who take full responsibility for the scientific content and integrity of the manuscript.

\newpage

\let\oldaddcontentsline\addcontentsline
\renewcommand{\addcontentsline}[3]{}

\bibliography{main}

\let\addcontentsline\oldaddcontentsline

\clearpage
\newpage

\clearpage
\onecolumngrid

% ============================================================
% Supplemental Material numbering
% ============================================================

% Reset every independent counter.
\setcounter{section}{0}
\setcounter{equation}{0}
\setcounter{figure}{0}
\setcounter{table}{0}

\setcounter{theorem}{0}
\setcounter{lemma}{0}
\setcounter{proposition}{0}
\setcounter{corollary}{0}
\setcounter{definition}{0}
\setcounter{problem}{0}

% Printed numbering in the Supplemental Material.
\renewcommand{\thesection}{S\arabic{section}}
\renewcommand{\theequation}{S\arabic{equation}}
\renewcommand{\thefigure}{S\arabic{figure}}
\renewcommand{\thetable}{S\arabic{table}}

\renewcommand{\thetheorem}{S\arabic{theorem}}
\renewcommand{\thelemma}{S\arabic{lemma}}
\renewcommand{\theproposition}{S\arabic{proposition}}
\renewcommand{\thecorollary}{S\arabic{corollary}}
\renewcommand{\thedefinition}{S\arabic{definition}}
\renewcommand{\theproblem}{S\arabic{problem}}

% Unique PDF hyperlink destinations.
% These prevent conflicts with counters used in the main text.
\renewcommand{\theHsection}{supp.section.\arabic{section}}
\renewcommand{\theHequation}{supp.equation.\arabic{equation}}
\renewcommand{\theHfigure}{supp.figure.\arabic{figure}}
\renewcommand{\theHtable}{supp.table.\arabic{table}}

\renewcommand{\theHtheorem}{supp.theorem.\arabic{theorem}}
\renewcommand{\theHlemma}{supp.lemma.\arabic{lemma}}
\renewcommand{\theHproposition}{supp.proposition.\arabic{proposition}}
\renewcommand{\theHcorollary}{supp.corollary.\arabic{corollary}}
\renewcommand{\theHdefinition}{supp.definition.\arabic{definition}}
\renewcommand{\theHproblem}{supp.problem.\arabic{problem}}

\begin{center}
  \textbf{\large
  Modular-Annihilator Parent Hamiltonians for Purified Gibbs States: Spectral Design
  and Controlled Approximation: Supplemental Material}
\end{center}

\bigskip

\section{Literature Review}
\label{sec:gibbs_sampler_review}

Preparing Gibbs states is a fundamental task in quantum simulation, quantum many-body physics, and quantum algorithms. A natural approach is to engineer an open-system dynamics whose stationary state is the desired Gibbs state. If this dynamics converges sufficiently rapidly, simulating it provides an algorithm for thermal-state preparation. Such methods are commonly referred to as dissipative quantum Gibbs samplers.

The study of dissipative Gibbs sampling originates from the theory of open quantum systems. In the weak system-bath coupling limit, Davies derived a Markovian master equation that relaxes a quantum system toward thermal equilibrium \cite{Davies1974Markovian}. Davies generators therefore provide one of the canonical physical models of quantum Gibbs sampling. Their equilibrium properties follow from the thermal properties of the bath and an appropriate quantum detailed-balance condition.

Quantum detailed balance plays a role analogous to detailed balance in classical Markov-chain Monte Carlo. Unlike in the classical case, however, several inequivalent notions of quantum detailed balance exist. Two particularly important notions are GNS and KMS detailed balance \cite{Alicki1976DetailedBalance,Kossakowski1977Quantum}. GNS detailed balance is the stronger condition and is satisfied by Davies generators, while KMS detailed balance defines a broader class of thermalizing dynamics. These symmetry properties make it possible to analyze convergence using spectral gaps, Poincar'e inequalities, and quantum logarithmic Sobolev inequalities \cite{Kastoryano2013Quantum}. 

For commuting Hamiltonians, dissipative Gibbs sampling is relatively well understood. Kastoryano and Brandao established a connection between spatial clustering of Gibbs states and rapid mixing of the corresponding Gibbs samplers \cite{Kastoryano2016Quantum}. For generic noncommuting many-body Hamiltonians, the Davies construction faces an important difficulty. It relies on resolving individual transition frequencies of the Hamiltonian. In a many-body system, energy levels may become exponentially dense as the system size increases, so resolving all relevant frequencies can require extremely long evolution times. Consequently, although Davies dynamics has a clear physical interpretation and exact thermal fixed point, it is generally not an efficient algorithmic construction for large interacting quantum systems.

An early attempt to overcome this difficulty was the quantum Metropolis algorithm proposed by Temme \emph{et al.} \cite{Temme2011QuantumMetropolis}. This algorithm quantizes the classical Metropolis acceptance-rejection procedure using phase estimation and quantum rewinding. It established an important conceptual connection between quantum Gibbs preparation and Markov-chain Monte Carlo. However, its reliance on coherent energy estimation leads to significant implementation difficulties. This motivated the development of continuous-time dissipative approaches that avoid precise measurement of the many-body spectrum.

A major recent development is the construction of Gibbs samplers based on smoothly filtered jump operators \cite{chen2023an,chen2025efficient}. 
Chen \emph{et al.} developed an exact and efficiently implementable KMS-balanced Gibbs sampler for general noncommuting Hamiltonians. Their construction can be interpreted as a continuous-time noncommutative analogue of classical Metropolis dynamics. For geometrically local Hamiltonians, Lieb--Robinson bounds imply that the resulting dissipative terms are quasi-local. Thus, at fixed temperature and accuracy, the sampler can be approximated using operators supported only on finite neighborhoods rather than requiring access to the full many-body spectrum.

The general structure of KMS-balanced quantum Gibbs samplers was further studied by Ding \emph{et al.} \cite{ding2025efficient}. They identified a broad family of filtered jump operators that satisfy KMS detailed balance and showed that only a finite number of jump operators is required in principle. Their framework makes clear that imposing detailed balance and obtaining rapid convergence are conceptually different tasks. Many different Lindbladians can share the same Gibbs stationary state, but their mixing properties may be very different.

Discrete-time versions of these ideas have also been developed. Gilyen \emph{et al.} introduced quantum generalizations of Glauber and Metropolis dynamics \cite{Gilyen2024QuantumGlauberMetropolis}. These constructions translate continuous-time Gibbs samplers into efficiently implementable quantum Markov chains while retaining detailed balance, locality, and useful spectral properties. This further strengthens the analogy between dissipative quantum Gibbs sampling and classical Markov-chain Monte Carlo.

An important recent question concerns the physical origin of KMS-balanced Gibbs samplers. Although they were initially motivated primarily as algorithms, Scandi and Alhambra showed that KMS-balanced dynamics can also emerge from microscopic system-bath interactions \cite{Scandi2026Thermalization}. Their derivation avoids the full rotating-wave approximation used in the Davies limit and remains applicable when the many-body energy spectrum is very dense. In this picture, KMS detailed balance naturally appears when nearby transition frequencies are not resolved completely. The resulting dynamics remains quasi-local and can be efficiently simulated. In an appropriate limit, it reduces to the standard Davies generator. This provides a physical interpretation of KMS Gibbs samplers beyond their algorithmic usefulness.

The remaining central challenge is the mixing time. Efficiently simulating a Lindbladian does not automatically imply efficient Gibbs-state preparation. One must also show that the dissipative dynamics approaches its stationary Gibbs state sufficiently quickly. For general noncommuting Hamiltonians, proving strong lower bounds on the Lindbladian spectral gap or on suitable logarithmic Sobolev constants remains difficult. Recent work has nevertheless established rapid mixing in several important regimes. Rouze \emph{et al.} proved efficient thermalization for a broad family of local and sufficiently rapidly decaying Hamiltonians at high temperature \cite{rouze2026efficient}. Their analysis maps the KMS-symmetric Lindbladian to a Hermitian parent operator and treats finite temperature as a perturbation of the infinite-temperature problem. They also showed that low-temperature dissipative dynamics can be computationally complicated, demonstrating that rapid mixing cannot be expected universally.

Overall, dissipative Gibbs sampling has developed from the traditional theory of Davies thermalization into a quantum-algorithmic analogue of classical Markov-chain Monte Carlo. A central recent insight is that exact Gibbs fixed points, efficient implementation, and quasi-locality can coexist even for generic noncommuting Hamiltonians. With the properties of such samplers becoming increasingly well understood, dissipative Gibbs sampling is emerging as one of the most physically-meaningful applications of quantum devices.

\section{Preliminaries}

\subsection{Matrix vectorization}

Suppose $\caH$ is a finite dimensional Hilbert space. Let $\{|e_i\>\}$ denote the standard basis for $\caH$, and $\caB(\caH)$ be the set of bounded operators on the space. We introduce a linear superoperator vec : $\caB(\caH) \to \caH \otimes \caH$ as follows: 
\begin{equation}
    \mathrm{vec}(|e_i\>\<e_j|) = |e_i\>\otimes |e_j\>\qquad \forall |e_i\>, |e_j\>.
\end{equation}
Consider $|l\> = \sum_i l_i |e_i\>, |r\> = \sum_{i} r_i|e_i\>$. By definition, we have
\begin{equation}
    \mathrm{vec}(|l\>\<r|) = \sum_{i,j}l_i r_j^*\mathrm{vec}(|e_i\>\<e_j|) = \sum_{i,j}l_i r_j^*|e_i\>\otimes|e_j\> = |l\>\otimes |r^*\>.
\end{equation}
Hence, for any Hermitian operator with decomposition $H = \sum_i \lambda_i|\psi_i\>\<\psi_i|$, its vectorization is
\begin{equation}
    \mathrm{vec}(H) = \sum_i \lambda_i|\psi_i\>\otimes|\psi_i^*\>.
\end{equation}
Under this framework, for any $A,B,C\in \caB(\caH)$, we have
\begin{equation}
    \mathrm{vec}(ABC) = (A \otimes C^\top)\mathrm{vec}(B).
\end{equation}
Here is a quick proof:
\begin{equation}
    \mathrm{vec}(ABC) = \sum_{ij}B_{ij}\mathrm{vec}(A|e_i\>\<e_j|C) = \sum_{ij}B_{ij}(A\otimes C^\top)|e_i\>|e_j\> = (A\otimes C^\top)\mathrm{vec}(B).
\end{equation}
Using this principle, we can map any superoperator to a matrix.

\subsection{Bohr frequency decomposition}

Suppose $H$ is a finite dimensional Hamiltonian. Denote its spectrum and spectral decomposition as
\begin{equation}
    \mathrm{spec}(H) \coloneqq \{E : \det(H - EI) = 0\},\qquad H = \sum_n E_n|\psi_n\>\<\psi_n|.
\end{equation}
Then the set of Bohr frequencies is
\begin{equation}
    \caF(H) \coloneqq \{\nu = E_1 - E_2: E_1,E_2\in \mathrm{spec}(H)\}.
\end{equation}
For each $\nu\in\caF(H)$, we define a superoperator:
\begin{equation}
    \Pi_\nu(X) \coloneqq \sum_{(m,n): E_m - E_n = \nu}\<\psi_m|X|\psi_n\>|\psi_m\>\<\psi_n| = \lim_{T\to\infty}\frac{1}{2T}\int_{-T}^T dt e^{\rmi H t}X e^{-\rmi Ht} e^{-\rmi\nu t}.
\end{equation}
Every operator $X$ can then be decomposed as
\begin{equation}
    X = \sum_{\nu\in\caF(H)}\Pi_\nu(X).
\end{equation}
We abbreviate $\Pi_\nu(X)$ as $X_\nu$ henceforth. Observe that
    $\Pi_\nu(X)^\dag = \Pi_{-\nu}(X^\dag)$.
Therefore, when $X$ is Hermitian, we have $X_\nu^\dag = X_{-\nu}$.

In the following sections, we will frequently use the modular transformation:
\begin{equation}
    \Delta(z,X) \coloneqq e^{z H}X e^{-z H}.
\end{equation}
By definition, we have
\begin{equation}
    [H,X_\nu] = \nu X_\nu,\qquad \Delta(z,X_\nu) = e^{z \nu}X_\nu,\qquad \Delta(z,X) = \sum_{\nu\in \caF(H)}e^{z\nu}X_\nu.
\end{equation}

\subsection{Quantum Gibbs sampling via Lindbladian}
\label{sec:lindbladian}

Given a Hamiltonian $H$ and an inverse temperature $\beta$, the purpose of quantum Gibbs sampling is to prepare the Gibbs state $\rho_\beta$ defined by:
\begin{equation}
    \rho_\beta \coloneqq \frac{e^{-\beta H}}{\Tr(e^{-\beta H})}.
\end{equation}
Currently, a well-known method for quantum Gibbs sampling is through the simulation of Lindbladian evolution. Let $\caL$ be a superoperator, the solution of
    $d\rho(t)/dt = \caL[\rho(t)]$
is denoted as $e^{\caL t}[\rho(0)]$. In Lindbladian evolution, the superoperator writes
\begin{equation}
    \caL(\cdot) = -\rmi[C,\cdot] + \sum_j \left[L_j (\cdot) L_j^\dag  - \frac{1}{2}\{L_j^\dag L_j,\cdot\}\right],
\end{equation}
where $-\rmi[C,\cdot]$ is called the coherent term, and $\{L_j\}$ is the set of jump operators. If the Lindbladian $\caL$ generates a primitive quantum Markov semigroup, then there exists a unique full-rank steady state ($\rho_\beta$ in our context) such that
\begin{equation}
    \lim_{t\to\infty} e^{\caL t}(\rho) = \rho_\beta\qquad \forall \rho.
\end{equation}
The Lindbladian evolution can be exactly solvable by matrix vectorization. The matrix form of the Lindbladian writes
\begin{equation}
    \mathrm{M}_\caL = -\rmi C \otimes I + \rmi I \otimes C^\top + \sum_j \left(L_j \otimes L_j^* - \frac{1}{2}L_j^\dag L_j \otimes I - \frac{1}{2}I \otimes L_j^\top L_j^*\right).
\end{equation}

The matrix itself is not Hermitian. However, with the KMS detailed-balanced condition is fulfilled, we can construct a Hermitian superoperator through similarity transformation
\begin{equation}
    \caH(\cdot) \coloneqq -\rho_\beta^{-\frac{1}{4}}\caL\left[\rho_\beta^{\frac{1}{4}} (\cdot) \rho_\beta^{\frac{1}{4}}\right]\rho_\beta^{-\frac{1}{4}}.
\end{equation}
The Hamiltonian satisfies $\rmM_\caH|\rb^{1/2}\> = 0$.

Let
\begin{equation}
    \mathrm{M}_\rho \coloneqq \rho_\beta^{\frac{1}{4}} \otimes \rho_\beta^{\frac{1}{4},\top},\qquad \Delta(\cdot) \coloneqq \rho_\beta^{-\frac{1}{4}}(\cdot) \rho_\beta^{\frac{1}{4}},
\end{equation}
then the matrix form of $\caH$ becomes
\begin{align}\label{eq:parentH}
    \mathrm{M}_\caH &= -\mathrm{M}_\rho^{-1} \mathrm{M}_\caL \mathrm{M}_\rho.
\end{align}
After simplification, 
\begin{gather}
\rmi \Delta(C) \otimes I - \rmi I \otimes\Delta^{-1}(C)^\top - \sum_j \left\{\Delta(L_j) \otimes \Delta^{-1}(L_j^\dag)^\top - \frac{1}{2}\Delta(L_j^\dag L_j) \otimes I - \frac{1}{2}I \otimes \Delta^{-1}(L_j^\dag L_j)^\top\right\}.
\end{gather}
The Lindbladian spectral gap is defined by
\begin{equation}
    \mathrm{Gap}(\caL) \coloneqq \inf_{\Tr(A^\dag \rb) = 0} \frac{-\Tr(A^\dag \sqrt{\rb} \caL^\dag(A)\sqrt{\rb})}{\Tr(A^\dag \sqrt{\rb} A \sqrt{\rb})}.
\end{equation}
In matrix vectorization, we obtain
\begin{gather}
    \Tr(A^\dag \sqrt{\rb} A \sqrt{\rb}) = \<A|\sqrt{\rb} \otimes \sqrt{\rb^\top}|A\> = \<A|\rmM_\rho^2|A\>,\\
    \Tr(A^\dag \sqrt{\rb} \caL^\dag(A)\sqrt{\rb}) = \<A|\sqrt{\rb} \otimes \sqrt{\rb^\top}\rmM_\caL^\dag|A\> = \<A|\rmM_\rho^2 \rmM_\caL^\dag|A\>.
\end{gather}
Since $\rmM_\caH$ is Hermitian,
\begin{equation}
    \rmM_\caH = -\rmM_\rho \rmM_\caL^\dag \rmM_\rho^{-1}.
\end{equation}
Therefore,
\begin{equation}
    \<A|\rmM_\rho^2 \rmM_\caL^\dag|A\> = -\<A|\rmM_\rho  \rmM_\caH \rmM_\rho |A\>.
\end{equation}
Let $|A'\> = \rmM_\rho |A\>$, then the Lindbladian spectral gap becomes
\begin{equation}
    \mathrm{Gap}(\caL) = \inf_{\<A'|\sqrt{\rb}\> = 0}\frac{\<A'|\rmM_\caH|A'\>}{\<A'|A'\>},
\end{equation}
which is exactly the spectral gap of the parent Hamiltonian $\rmM_\caH$ (more details). 

The spectral gap is directly related to the mixing time, as shown in the next lemma.
\begin{lemma}Suppose $\caL$ has a unique full rank steady state and satisfies the KMS detailed-balance condition. Then the mixing time has upper bound:
\begin{equation}
    t_{\mathrm{mix}}(\varepsilon) = O\left(\frac{\log(\varepsilon^{-1}) + \log(\|\rb^{-1}\|)}{\mathrm{Gap}(\caL)}\right).
\end{equation}
\end{lemma}
\begin{proof}
    Recall that the KMS inner product is defined by inner product
    \begin{equation}
        \<A,B\>_\beta = \tr(A^\dag \rho_\beta^{\frac{1}{2}}B \rho_\beta^{\frac{1}{2}}).
    \end{equation}
    Suppose the initial state is $\sigma$. Introduce
    \begin{equation}
    A(t)\coloneqq
    e^{\caL^\dagger t}
    \left(\rb^{-1/2}\sigma\rb^{-1/2}\right)-I.
    \end{equation}
    The KMS detailed-balance relation implies
    \begin{equation}
    \rb^{1/2}A(t)\rb^{1/2}
    =e^{\caL t}(\sigma)-\rb .
    \end{equation}
    Since $\caL^\dagger(I)=0$, we have
    $dA(t)/dt=\caL^\dagger(A(t))$ and
    $\tr[A(t)\rb]=0$.
    
    By virtue of Holder's inequality, we have
    \begin{equation}
        \|e^{\caL t}(\sigma) - \rb\|_1 = \|\rb^{\frac{1}{2}}A(t)\rb^{\frac{1}{2}}\|_1 \le \|\rb^{\frac{1}{4}}\|_4^2 \cdot \|\rb^{\frac{1}{4}}A(t)\rb^{\frac{1}{4}}\|_2 = \<A(t),A(t)\>_\beta^{\frac{1}{2}} \le \<A(0),A(0)\>_\beta^{\frac{1}{2}} e^{-\mathrm{Gap}(\caL)t}.
    \end{equation}
    Observe that $\<A(0),A(0)\>_\beta \le 4\|\rb^{-1}\|$. Finally,
    \begin{equation}
        \|e^{\caL t}(\sigma) - \rb\|_1 \le 2\|\rb^{-1}\|^{\frac{1}{2}} e^{-\mathrm{Gap}(\caL)t}.
    \end{equation}
    The lemma is proved by requiring $2\|\rb^{-1}\|^{\frac{1}{2}} e^{-\mathrm{Gap}(\caL)t} \le \varepsilon$.
\end{proof}

\subsubsection{Davies generator}

Davies generator \cite{Davies1974Markovian} is an important type of quantum Markov semigroup. The construction starts with the Bohr frequency decomposition of an operator $J$:
\begin{equation}
    J = \sum_{\nu\in\caF(H)}J_\nu.
\end{equation}
The Gibbs state $\rb$ is the steady state of the following Lindbladian:
\begin{equation}
    \caL'_{\mathrm{Davies}}(\cdot) \coloneqq \sum_{\nu\in\caF(H)}\eta_\nu \left[J_\nu (\cdot) J_\nu^\dag - \frac{1}{2}\{J_\nu^\dag J_\nu,\cdot\}\right],\qquad \eta_\nu e^{\beta \nu} = \eta_{-\nu}.
\end{equation}
Let $\caA$ denote a set of such operators. Then a general Davies generator is in the form of
\begin{equation}
    \caL_{\mathrm{Davies}}(\cdot) = \sum_{J_a\in\caA}\sum_{\nu\in\caF(H)}\eta_\nu \left[J_{a,\nu} (\cdot) J_{a,\nu}^\dag - \frac{1}{2}\{J_{a,\nu}^\dag J_{a,\nu},\cdot\}\right].
\end{equation}
Under the notion of \sref{sec:lindbladian}, the coherent term is $C = 0$, and the set of jump operators is $\{\sqrt{\eta_\nu} J_{a,\nu}\}$. Hence, the parent Hamiltonian is
\begin{equation}
    \rmM_{\caH,\mathrm{Davies}} = \frac{1}{2}\sum_{a\in\caA,\nu\in\caF(H)}\left(\sqrt{\eta_\nu} J_{a,\nu} \otimes I - I \otimes \sqrt{\eta_{-\nu}} J_{a,\nu}^\top\right)^\dag\left(\sqrt{\eta_\nu} J_{a,\nu} \otimes I - I \otimes \sqrt{\eta_{-\nu}} J_{a,\nu}^\top\right).
\end{equation}

\subsubsection{DLL Lindbladian}

The DLL Lindbladian is a general framework that describes a large class of Lindbladian with known steady states, where the number of dissipative terms is irrelevant to the size of $\caF(H)$. Note that the Chen-Kastoryano-Gilyen (CKG) Lindbladian \cite{chen2023an,chen2025efficient} belongs to this category.

The construction starts with a set of jump operators $L_a$, which satisfies
\begin{equation}
    \Delta_{\beta/2}(L_a) = L_a^\dag.
\end{equation}
The full DLL Lindbladian is
\begin{equation}
    \caL(\cdot) = -\rmi[C,\cdot] + \sum_a \left[L_a(\cdot) L_a^\dag  - \frac{1}{2}\{L_a^\dag L_a, \cdot\}\right],
\end{equation}
where
\begin{equation}
    C \coloneqq \frac{\rmi}{2}\sum_a \sum_{\nu\in\caF(H)}\tanh(\beta \nu/4)(L_a^\dag L_a)_\nu.
\end{equation}
The purpose of introducing $C$ is to ensure that 
\begin{equation}
    K = -\frac{1}{2}\sum_a L_a^\dag L_a - \rmi C,\qquad \Delta_{\beta/2}(K) = K^\dag.
\end{equation}
so that the Lindbladian satisfies the KMS detailed-balance condition. For the case of $[H,\sum_a L_a^\dag L_a] = 0$, we have $C = 0$.

The full parent Hamiltonian writes 
\begin{equation}
    \rmM_\caH = -\Delta_{\beta/4}(K) \otimes I - I \otimes \left(\Delta_{-\beta/4}(K^\dag)\right)^\top - \sum_a \left\{\Delta_{\beta/4}(L_a) \otimes \left(\Delta_{-\beta/4}(L_a^\dag)\right)^\top\right\}.
\end{equation}
In a recent paper \cite{leng2026accelerate}, the authors proved that this Hamiltonian can be written in the form of
\begin{equation}
    \rmM_\caH = \frac{1}{2}\sum_{a}\int_{-\infty}^{\infty} dt\, g(t)  \Gamma_a(t)^\dag \Gamma_a(t). 
\end{equation}
However, the integration over time makes the Hamiltonian difficult to encode. In this paper, we will show that in specific situations, we can write it as a discrete summation which only contains $O(N)$ number of terms, where $N$ is total number of qubits. 

\subsection{Majorana operators and free-fermion models}

The Majorana operators $\{\omega_a\}_{a=1}^{2N}$ satisfy
\begin{equation}
    \omega_a^\dag=\omega_a,\qquad
    \{\omega_a,\omega_b\}=2\delta_{ab}I .
\end{equation}

An $N$-mode free-fermion model can be written as
\begin{equation}
    H = \sum_{a,b=1}^{2N}h_{ab}\omega_a \omega_b,
\end{equation}
where $h$ is a purely imaginary Hermitian matrix (hence $h^* = h^\top = -h$), and $\{\omega_a\}$ is the set of Majorana operators. Note that
\begin{equation}
    [H,\omega_a] = -4\sum_{b=1}^{2N}h_{ab}\omega_b,\qquad \|H\| = \|h\|_1.
\end{equation}
For a fixed complex number $z$, we have
\begin{equation}
   e^{zH}\omega_a e^{-zH} = \sum_{b=1}^{2N}\left(e^{-4zh}\right)_{ab}\omega_b,
\end{equation}
and
\begin{equation}
    \Delta(\omega_a) = \sum_{b=1}^{2N}\left(e^{-\beta h}\right)_{ab}\omega_b.
\end{equation}

Finally, the Gibbs state of free-fermion model is exactly solvable. Since $h$ is purely imaginary Hermitian, there must exist an orthogonal matrix $R$, such that
\begin{equation}
    R^\top h R = \bigoplus_{n=1}^N \frac{\rmi}{2}\begin{pmatrix}
        0 & \lambda_n \\
        -\lambda_n & 0
    \end{pmatrix}.
\end{equation}
Let $\tilde{\omega}_a \coloneqq \sum_b R_{ba}\omega_b$. Then $\{\tilde{\omega}_a\}_{a=1}^{2N}$ is also a set of Majorana operators, and
\begin{equation}
    H = \rmi \sum_{n=1}^N \lambda_n \tilde{\omega}_{2n-1}\tilde{\omega}_{2n}.
\end{equation}
Let $\tilde{Z}_n \coloneqq -\rmi \tilde{\omega}_{2n-1}\tilde{\omega}_{2n}$. Then $[\tilde{Z}_n, \tilde{Z}_m] = 0$ for all $m,n$. The Gibbs state is 
\begin{equation}
    \rb \propto \prod_{n=1}^N e^{\beta \lambda_n \tilde{Z}_n} = \prod_{n=1}^N \left[\cosh(\beta \lambda_n)I + \sinh(\beta \lambda_n)\tilde{Z}_n\right].
\end{equation}
After normalization, we obtain
\begin{equation}\label{eq:gibbs_ff}
    \rb = \prod_{n=1}^N \frac{I + \tanh(\beta \lambda_n)\tilde{Z}_n}{2}.
\end{equation}

\subsection{Irreducible unital algebra}

Let $\caW$ be a finite-dimensional complex Hilbert space, and let
$\caB(\caW)$ denote the algebra of linear operators on $\caW$. Let $\{J_a\}$ be a finite set of linear operators in $\caB(\caW)$, then the unital algebra generated by $\{J_a\}$ is the smallest set of linear operators that contains $I \cup\{J_a\}$, closed under linear combinations and multiplications. The unital algebra can be denoted by $\mathrm{Alg}(\{J_a\})$.

A unital algebra $\caA\subseteq \caB(\caW)$ is called irreducible if the only subspaces
$\caV\subseteq \caW$ satisfying
\begin{equation}
    O_a\caV\subseteq \caV,\qquad \forall O_a\in\caA,
\end{equation}
are $\caV=\{0\}$ and $\caV=\caW$.

By Burnside's theorem, in finite dimension over $\mathbb C$, a unital algebra
$\caA\subseteq \caB(\caW)$ is irreducible iff $\caA=\caB(\caW). $
Therefore, if the linear span of all finite products of generators $\{J_a\}$ equals
$\caB(\caW)$, then the algebra generated by $\{J_a\}$ is irreducible.

A prominent example on the $N$-qubit Hilbert space (denoted by $(\bbC^2)^{\otimes N}$) is given by the Majorana operators.
\begin{lemma}
The algebra generated by the Majorana operators
$\{\omega_a\}_{a=1}^{2N}$ is irreducible on $(\bbC^2)^{\otimes N}$.
\end{lemma}

\begin{proof}
The products of Majorana operators generate all Pauli strings. Indeed,
\begin{equation}
    Z_a = -\rmi\omega_{2a-1}\omega_{2a},
\end{equation}
and we can recover $X_a$ and $Y_a$ from products of
$\omega_b$ and $Z_b$. Hence the algebra generated by
$\{\omega_a\}_{a=1}^{2N}$ contains the full Pauli operators. Since Pauli strings form a basis of $\caB((\bbC^2)^{\otimes N})$, by Burnside's theorem, it is irreducible.
\end{proof}

\subsection{Krylov subspace approximation}
\label{app:lanczos}

Let $G$ be a Hermitian linear operator acting on a finite-dimensional vector space equipped with inner product $\<x,y\>$, and the vector, matrix norm are defined by
\begin{equation}
    \|x\| \coloneqq \sqrt{\<x,x\>},\qquad \|M\| \coloneqq \max_{\|x\| = 1}\|Mx\|.
\end{equation}
Suppose $v_1$ is a unit vector. Define the $m$-th order Krylov space as
\begin{equation}
    \caK_m(G,v_1) \coloneqq \mathrm{span}\{v_1, G v_1, G^2v_1,\ldots, G^{m-1} v_1\}.
\end{equation}
The Lanczos recursion procedure creates an orthogonal basis $\{v_1,v_2,\ldots,v_m\}$
of $\caK_m(G,v_1)$ together with a tridiagonal representation of $G$ on this subspace. 

Starting with $b_0 = 0$ and $n = 1$, the Lanczos recursion is performed by
\begin{enumerate}
    \item $w_n = G v_n - b_{n-1}v_{n-1}$,
    \item $a_n = \langle v_n,w_n\rangle$,
    \item $w_n \leftarrow w_n-a_nv_n$,
    \item $b_n = \|w_n\|$,
    \item $v_{n+1} = w_n/b_n$.
\end{enumerate}
The procedure terminates when either $b_n=0$ or $n = m$. We assume $b_n \neq 0$ for all $0 < n < m$ for simplicity.

In the Lanczos basis, the action $G$ is represented by an $m\times m$ tridiagonal matrix
\begin{equation}
    T_m
    =
    \begin{pmatrix}
        a_1 & b_1 & 0 & \cdots \\
        b_1 & a_2 & b_2 & \cdots \\
        0 & b_2 & a_3 & \cdots \\
        \vdots & \vdots & \vdots & \ddots
    \end{pmatrix}.
\end{equation}
Let $\{t_1,t_2,\ldots,t_m\}$ denote a trivial vector basis, then $T_m$ can be written as
\begin{equation}
    T_m = \sum_{i=1}^m a_i t_i t_i^\top + \sum_{i=1}^{m-1}b_i (t_i t_{i+1}^\top + t_{i+1} t_i^\top).
\end{equation}
For any matrix function $f$, the Krylov approximation of $f(G)v_1$ is given by
\begin{equation}
    L_v(t_j) \coloneqq v_j,\quad L_v \left[f(T_m) t_1\right] = \sum_{j=1}^m \left[ f(T_m)\right]_{j1}v_j .
\end{equation}
The error bound $\|f(G)v_1 - L_v[f(T_m)e_1]\|$ can be obtained by the polynomial approximation method. In particular, when $f(G) = e^{xG}$, we have the following result.
\begin{lemma}\label{lem:polyapprox}
    Suppose $x > 0$. In order to have $\|e^{xG}v_1 - L_v( e^{xT_m}t_1)\| \le \varepsilon$, it suffices to take
    \begin{equation}
        m = O\left(x\|G\| + \log(\varepsilon^{-1})\right).
    \end{equation}
\end{lemma}

\begin{proof}
For any matrix function $f$, the Lanczos recursion is nearly-optimal in that
\begin{equation} \label{eq:polynomial_approximation}
    \|f(G)v_1 - L_v\left[f(T_m)e_1\right]\| \le 2
    \min_{p\in \bbP_{m-1}}
    \max_{\lambda\in[\lambda_{\min}(G),\lambda_{\max}(G)]}
    |f(\lambda)-p(\lambda)|,
\end{equation}
where $\bbP_{m-1}$ denotes the set of polynomials of degree at most $m-1$. Thus, each $p\in \bbP_{m-1}$ leads to an error bound. If we choose $p(\lambda)$ as the $(m-1)$-th order Taylor expansion of $f(\lambda)$, then
\begin{equation}
    \|f(G)v_1 - L_v\left[f(T_m)t_1\right]\| \le 2\max_{\lambda\in [\lambda_{\min}(G),\lambda_{\max}(G)]}\left|\sum_{k=m}^{\infty} \frac{1}{k!}f^{(k)}(0)\lambda^k\right|.
\end{equation}
With $f(\lambda) = e^{x\lambda}$, we further obtain
\begin{equation}
    \max_{\lambda\in[\lambda_{\min}(G),\lambda_{\max}(G)]}\left|\sum_{k=m}^{\infty}\frac{(x\lambda)^k}{k!}\right|
    \le
    \sum_{k=m}^{\infty}
    \frac{(x\|G\|)^k}{k!} \le e^{x\|G\|}
    \left(
        \frac{ex\|G\|}{m}
    \right)^m.
\end{equation}
Here we used the fact that $|\lambda| \le \|G\|$. Finally, the error in bounded by
\begin{equation}
    \left\|
        e^{xG}v_1 - L_v \left(e^{xT_m}e_1\right)
    \right\|
    \le
    2e^{x\|G\|}
    \left(
        \frac{ex\|G\|}{m}
    \right)^m.
\end{equation}
The RHS has upper bound $\varepsilon$ provided that
\begin{equation}
    m \ln\left(\frac{m}{ex\|G\|}\right) \ge \ln(2\varepsilon^{-1}) + x\|G\|.
\end{equation}
Let $C_1,C_2$ be two positive numbers. If $m \ge C_1\log(2\varepsilon^{-1}) + C_2 ex\|G\|$, then we get
\begin{equation}
    m \ln\left(\frac{m}{ex\|G\|}\right) \ge m \ln C_2 \ge C_1\ln (C_2)\log(2\varepsilon^{-1}) + C_2\ln (C_2) ex\|G\|.
\end{equation}
The lemma is proved by choosing $C_1 \ge 1, C_2 \ge e$.
\end{proof} 

\section{Proof of \lref{lemma:preconditioning}}
\label{app:lem1}

\begin{lemma}
[Restatement of \lref{lemma:preconditioning}]
Let $\{\Gamma_a\}_{a\in\caA}$ be a finite family of modular
annihilators with $\Gamma_a|\rb^{1/2}\> = 0$, and let
$G\in\bbC^{|\caA|\times|\caA|}$ be Hermitian and positive definite.
Then $\rmM_\caH(G) \succeq 0$ and $\rmM_\caH(G)|\rb^{1/2}\> = 0$.
Moreover,
\begin{equation}\label{eq:weighted_parent_bounds}
    \lambda_{\min}(G)\,\rmM_\caH(I)
    \preceq
    \rmM_\caH(G)
    \preceq
    \lambda_{\max}(G)\,\rmM_\caH(I).
\end{equation}
Consequently, if $\rmM_\caH(I)$ has a unique ground state, then
\begin{equation}\label{eq:weighted_gap_bounds}
    \lambda_{\min}(G)\,
    \mathrm{Gap}\!\left[\rmM_\caH(I)\right]
    \le
    \mathrm{Gap}\!\left[\rmM_\caH(G)\right]
    \le
    \lambda_{\max}(G)\,
    \mathrm{Gap}\!\left[\rmM_\caH(I)\right].
\end{equation}
\end{lemma}

\begin{proof}
Define the block operator
\begin{equation}
    B_\Gamma
    :=
    \sum_{a\in\caA}|a\>\otimes\Gamma_a .
\end{equation}
Then
\begin{equation}
    \rmM_\caH(G)
    =
    \frac{1}{2}
    B_\Gamma^\dag
    (G\otimes I)
    B_\Gamma ,\qquad \rmM_\caH(I)
    =
    \frac{1}{2}
    B_\Gamma^\dag
    B_\Gamma.
\end{equation}
Since $G$ is positive definite, $\rmM_\caH(G)$ is positive
semidefinite. Moreover, we can decompose $G$ as $R^\dag R$, so that for any state $|\psi\>$,
\begin{equation}
\begin{aligned}
    \<\psi|\rmM_\caH(G)|\psi\>
    &=
    \frac{1}{2}
    \left\|
        (R\otimes I)B_\Gamma|\psi\>
    \right\|^2 .
\end{aligned}
\end{equation}
Because $G \succ 0$, the matrix $R\otimes I$ is invertible, and $(R\otimes I)B_\Gamma|\psi\>$ vanishes iff
$B_\Gamma|\psi\>=0$, which is equivalent to
$\Gamma_a|\psi\>=0$ for every $a\in\caA$. 

Finally, note that
\begin{equation}
    \lambda_{\min}(G)I
    \le
    G
    \le
    \lambda_{\max}(G)I.
\end{equation}
Conjugating the inequality by $B_\Gamma$ gives
\eref{eq:weighted_parent_bounds}. When the common kernel of $\Gamma_a$ is
one-dimensional, \eref{eq:weighted_gap_bounds} follows from the
variational characterization of the spectral gap on the orthogonal
complement of the ground state.
\end{proof}

\section{Proof of \thref{thm:lindbladian}}
\label{app:prooflindblad}

\begin{theorem}[Restatement of \thref{thm:lindbladian}]
    Suppose $\mathrm{Alg}\left(\{J_a\}_{a\in\caA}\right) = \caB((\bbC^2)^N)$ and each $J_a$ is Hermitian. If
    \begin{equation} \label{eq:condition}
        \sum_{a\in\caA} J_a \otimes \left(\Delta^2(J_a)\right)^\top = \sum_{a\in\caA}\Delta^{-2}(J_a) \otimes J_a^\top. 
    \end{equation}
    Then
\begin{gather}
    \caL_{\mathrm{SoS}}(\cdot) \coloneqq \frac{1}{2}\sum_{a\in\caA}[J_a (\cdot) \Delta^{2}(J_a) + \Delta^{-2}(J_a) (\cdot) J_a] - \Delta^{-1}(K) (\cdot) - (\cdot) \Delta(K), \\
    K \coloneqq \frac{1}{2}\sum_{a\in\caA} \Delta(J_a)\Delta^{-1}(J_a)
\end{gather}
    is a primitive Lindbladian.
\end{theorem}

To prove it, we first introduce a standard result in quantum dynamical semigroup \cite{Lindblad1976} and an auxiliary lemma.
\begin{lemma}\label{lem:lindbladian}
    A finite dimensional superoperator $\caL$ is a Lindbladian iff it satisfies
    \begin{enumerate}
        \item $\caL$ is Hermitian-preserving;
        \item $\caL^\dag(I) = 0$;
        \item let $P_{\perp} \coloneqq I - |\Omega\>\<\Omega|/d$, and $\caJ(\caL) \coloneqq \sum_{i,j}\caL(|i\>\<j|)\otimes |i\>\<j|$ is the Choi matrix of $\caL$, then 
        \begin{equation}
            P_\perp \caJ(\caL)P_\perp \ge 0.
        \end{equation}
    \end{enumerate}
\end{lemma}

Then \thref{thm:lindbladian} is proved as follows.
\begin{proof}

 Introduce $|X\>\!\> \coloneqq (X\otimes I)|\Omega\>$. Due to the Hermiticity of $J_a$, we get $K = K^\dag$, and
 \begin{equation}
     \caJ(\caL_{\mathrm{SoS}}) = \frac{1}{2}\sum_{a} |J_a\>\!\> \<\!\<\Delta^{-2}(J_a)| + \frac{1}{2}\sum_a|\Delta^{-2}(J_a)\>\!\>\<\!\<J_a| - |\Delta^{-1}(K)\>\!\>\<\!\<I| - |I\>\!\>\<\!\<\Delta^{-1}(K)|.
 \end{equation}
 Then $\caL_{\mathrm{SoS}}$ is Hermitian preserving as $\caJ(\caL_{\mathrm{SoS}})^\dag = \caJ(\caL_{\mathrm{SoS}})$.
 
 For any operator $X$, we have
\begin{align}
    \tr(\caL_{\mathrm{SoS}}(X)) &= \frac{1}{2}\sum_{a;s=\pm 1}\tr(X\Delta^{s+1}(J_a)\Delta^{s-1}(J_a)) - \tr(X\Delta^{-1}(K)) - \tr(X\Delta(K)) = 0.
\end{align}
Thus, $\tr(\caL_{\rm SoS}^\dag(I)X) = 0$ for all $X$, which implies $\caL^\dag_{\rm SoS}(I) = 0$, and condition 2 is fulfilled as well.

 As to condition 3, since $|I\>\!\> = |\Omega\>$ by definition, the last two terms vanishes after projecting to $P_{\perp}$. Hence,
 \begin{equation}
     P_\perp\caJ(\caL_{\mathrm{SoS}})P_\perp = \frac{1}{2}\sum_{a} P_\perp|J_a\>\!\> \<\!\<\Delta^{-2}(J_a)|P_\perp + \frac{1}{2}\sum_a P_\perp|\Delta^{-2}(J_a)\>\!\>\<\!\<J_a|P_\perp.
 \end{equation}
 Let
 \begin{equation}
     \caJ_P \coloneqq \sum_a |J_a\>\!\>\<\!\<J_a|,\qquad \caJ_\Delta \coloneqq \rb^{\frac{1}{2}} \otimes \rb^{-\frac{1}{2},\top}.
 \end{equation}
 Note that $\caJ_\Delta|X\>\!\> = |\Delta^{-2}(X)\>\!\>$ for all $X$. Hence,
 \begin{equation} \label{eq:commute}
     \sum_{a}|J_a\>\!\>\<\!\<\Delta^{-2}(J_a)| = \caJ_P \caJ_\Delta^\dag = \caJ_P \caJ_\Delta,\qquad \sum_a|\Delta^{-2}(J_a)\>\!\>\<\!\<J_a| = \caJ_\Delta \caJ_P,
 \end{equation}
 and
 \begin{equation}
    P_\perp \caJ(\caL_{\mathrm{SoS}})P_\perp = \frac{1}{2}P_\perp (\caJ_\Delta \caJ_P + \caJ_P\caJ_\Delta)P_\perp.
\end{equation}
 By virtue of \eref{eq:condition} and (\ref{eq:commute}), the two operators commute $[\caJ_P,\caJ_\Delta] = 0$. Since $\caJ_P \ge 0, \caJ_\Delta \ge 0$, we have $P_\perp \caJ(\caL_{\mathrm{SoS}})P_\perp \ge 0$ as well. This confirms condition 3 and completes the proof of the theorem.

Now we prove that because
$\mathrm{Alg}(\{J_a\}_{a\in\caA})$ is irreducible, 
$\caL_{\rm SoS}$ is primitive. That is, $\rb$ is its unique stationary state and
$\lim_{t\to\infty}e^{t\caL_{\rm SoS}}(\sigma)=\rb$ for every density
operator $\sigma$. 

By \eref{eq:parentH}, $\rmM_\caL=-\rmM_\rho\rmM_\caH\rmM_\rho^{-1}$ with
$\rmM_\rho=\rb^{1/4}\otimes\rb^{1/4,\top}\succ0$. Hence $\rmM_\caL$ is
similar to the Hermitian negative-semidefinite matrix $-\rmM_\caH$, so it
is diagonalizable with real nonpositive spectrum; in particular it has no
nonzero peripheral eigenvalue and no Jordan block. Since
$\<v|\rmM_\caH|v\>=\frac12\sum_a\|\Gamma_a|v\>\|^2$, we have
$\ker\rmM_\caH=\bigcap_a\ker\Gamma_a=\mathrm{span}\{|\rb^{1/2}\>\}$ by
\thref{thm:algebra}, and $\rmM_\rho|\rb^{1/2}\>=|\rb\>$. Therefore
$\ker\rmM_\caL=\mathrm{span}\{|\rb\>\}$ and the zero eigenvalue is simple.
Consequently $\lim_{t\to\infty}e^{t\rmM_\caL} = \Pi_0$ with $\mathrm{rank}\,\Pi_0=1$. For a
density operator $\sigma$, $\Pi_0(\sigma)=c\,\rb$, and trace preservation
forces $c=1$.
\end{proof}

\section{Proof of \coref{coro:gkls}}
\label{app:proof_coro}
Recall that 
\begin{equation}
    J_a = \sum_{\nu\in\caF(H)}J_{a,\nu},\qquad \Delta(J_a) = \sum_{\nu\in\caF(H)}e^{\beta \nu/4}J_{a,\nu}.
\end{equation}
Introduce the frequency-pair tensors \begin{equation}\label{eq:Tnumu} 
    T_{\nu\mu} \coloneqq \sum_{a\in\caA}J_{a,\nu}\otimes\bigl(J_{a,\mu}\bigr)^\top \qquad \nu,\mu\in \caF(H). 
\end{equation}  
Hence, the commutation condition can be translated into the following constraint.
\begin{lemma}\label{lemma:bohr-condition} Suppose condition \eref{eq:condition} hold. Then $T_{\nu\mu} = 0$ for all $\nu \neq -\mu$.
\end{lemma}
\begin{proof}
    Note that
    \begin{equation}
        \sum_{a\in\caA}J_a\otimes J_a^\top = \sum_{\nu,\mu\in\caF(H)}T_{\nu\mu}.
    \end{equation}
    The condition \eref{eq:condition} is equivalent to
    \begin{equation}\label{eq:1}
        \sum_{\nu,\mu\in \caF(H)}\left(e^{-\beta\nu/2} - e^{\beta\mu/2}\right)T_{\nu\mu} = 0.
    \end{equation}
    We first show that tensors associated with distinct frequency pairs are mutually orthogonal with respect to the Hilbert--Schmidt inner product. Indeed,
\begin{equation}
[H\otimes I,T_{\nu\mu}] = \nu T_{\nu\mu},
\end{equation}
whereas
\begin{equation}
[I\otimes H^\top,T_{\nu\mu}]
=
-\mu T_{\nu\mu}.
\end{equation}
Thus, \(T_{\nu\mu}\) is a joint eigenoperator of the commuting
Hermitian superoperators
\(\operatorname{ad}_{H\otimes I}\) and
\(\operatorname{ad}_{I\otimes H^\top}\), with joint eigenvalue
\((\nu,-\mu)\). Consequently,
\begin{equation}
\tr \left(
T_{\nu\mu}^\dagger T_{\nu'\mu'}
\right)
=
0
\qquad
\text{whenever }
(\nu,\mu)\neq(\nu',\mu').
\end{equation}
Take the Hilbert--Schmidt inner product of
\eref{eq:1} with \(T_{\nu\mu}\). Then for all $T_{\nu\mu} \neq 0$, we require that $e^{-\beta\nu/2} = e^{\beta\mu/2}$, which implies $\nu = - \mu$. Conversely, $\nu \neq -\mu$ implies $T_{\nu\mu} = 0$. This completes the proof of the lemma. 
\end{proof}

\begin{corollary} [Restatement of \coref{coro:gkls}] Suppose $\mathrm{Alg}\left(\{J_a\}_{a\in\caA}\right) = \caB((\bbC^2)^N)$ and each $J_a$ is Hermitian.  If $\caL_{\rm SoS}$ satisfies \eref{eq:condition}, then
\begin{equation}\caL_{\rm SoS}(\cdot) =\sum_{a\in\caA}\sum_{\nu\in\caF(H)} e^{-\beta\nu/2}\left[ J_{a,\nu}(\cdot)J_{a,\nu}^\dag -\frac{1}{2}\bigl\{J_{a,\nu}^\dag J_{a,\nu},\cdot\bigr\} \right]. 
\end{equation} 
\end{corollary} 
\begin{proof} 
Note that
\begin{align} 
    \frac{1}{2}\sum_{a\in\caA}\Bigl[J_a(\cdot)\Delta^2(J_a) +\Delta^{-2}(J_a)(\cdot)J_a\Bigr] =\frac{1}{2}\sum_{a\in\caA,\nu,\mu\in\caF(H)}\left(e^{\beta\mu/2} + e^{-\beta\nu/2}\right)J_{a,\nu}(\cdot)J_{a,\mu}
\end{align} 
By virtue of \pref{lemma:bohr-condition}, only the sectors with $\mu=-\nu$ survive in the expansion. Hence,
\begin{align} 
    \frac{1}{2}\sum_{a\in\caA}\Bigl[J_a(\cdot)\Delta^2(J_a) +\Delta^{-2}(J_a)(\cdot)J_a\Bigr] = \sum_{a\in\caA,\nu\in\caF(H)} e^{-\beta\nu/2}J_{a,\nu}(\cdot)J_{a,\nu}^\dag,
\end{align} 
where we used $J_a = J_a^\dag$ and $J_{a,-\nu}=J_{a,\nu}^\dag$. After computation,  
\begin{equation} 
K=\frac12\sum_{a\in\caA}\Delta(J_a)\Delta^{-1}(J_a) =\frac12\sum_{a\in\caA,\nu\in\caF(H)}e^{\beta\nu/2}J_{a,\nu}J_{a,\nu}^\dag =\frac12\sum_{a\in\caA,\nu\in\caF(H)}e^{-\beta\nu/2}J_{a,\nu}^\dag J_{a,\nu}.\end{equation} 
This completes the proof of the corollary.
\end{proof} 

\section{Verification of the counterexample in \sref{sec:map}}
\label{app:verification}

It is direct to show that $\caL_{\rm SoS}$ is Hermitian-preserving and $\caL_{\rm SoS}^\dag(I) = 0$. According to \lref{lem:lindbladian}, we only need to show that $P_\perp \caJ(\caL_{\rm SoS}) P_\perp \succeq 0$. In particular, for the SoS superoperator, we have
\begin{equation}
    P_\perp \caJ(\caL_{\rm SoS})P_\perp = \frac{1}{2}P_\perp (\caJ_{\Delta} \caJ_P + \caJ_P \caJ_\Delta)P_\perp.
\end{equation}

Recall that
\begin{equation}
    J_1 = \begin{bmatrix}
        1 & 1 \\
        1 & -1
    \end{bmatrix},\quad J_2 = \begin{bmatrix}
        0 & -\rmi \\
        \rmi & 0
    \end{bmatrix},\quad J_3 = \begin{bmatrix}
        1 & 0 \\
        0 & -1
    \end{bmatrix}.
\end{equation}
Hence, in the basis of $\{|00\>,|01\>,|10\>,11\>\}$,
\begin{gather}
    |J_1\>\!\> = (1,1,1,-1),\quad
    |J_2\>\!\> = (0,-\rmi,\rmi,0),\quad
    |J_3\>\!\> = (1,0,0,-1),\\
    \caJ_P = \sum_{a=1}^3|J_a\>\!\>\<\!\<J_a| = \begin{bmatrix}
        2 & 1 & 1 & -2 \\
        1 & 2 & 0 & -1 \\
        1 & 0 & 2 & 1 \\
        -2 & -1 & -1 & 2
    \end{bmatrix}.
\end{gather}
The choice of $\beta,H$ gives
\begin{equation}
   \caJ_\Delta  = \mathrm{diag}\{1,2,1/2,1\}.
\end{equation}
Therefore,
\begin{equation}
    P_\perp (\caJ_{\Delta} \caJ_P + \caJ_P \caJ_\Delta)P_\perp = \begin{bmatrix}
        4 & 3 & 0 & -4 \\
        3 & 8 & 0 & -3 \\
        0 & 0 & 2 & 0 \\
        -4 & -3 & 0 & 4
    \end{bmatrix}.
\end{equation}
Its eigenvalues are $0,2,8\pm 3\sqrt{2}$. Hence, $P_\perp \caJ(\caL_{\rm SoS}) P_\perp \succeq 0$, and the example is a Lindbladian.

\section{Proof of \thref{thm:optimal_generators} and \pref{prop:replacer}}

\label{app:replace}

\begin{theorem}[Restatement of \thref{thm:optimal_generators}]
Let $\{W_a\}_{a=1}^{d^2}$ be any Hermitian basis of
$\caB((\bbC^2)^{\otimes N})$ orthonormal in the Hilbert--Schmidt inner
product, $\Tr(W_a W_b)=\delta_{ab}$, and set
\begin{equation}\label{eq:optimal_J}
    J_a = e^{-\beta H/4}W_ae^{-\beta H/4}.
\end{equation}
Then each $J_a$ is Hermitian, $\mathrm{Alg}(\{J_a\})=\caB((\bbC^2)^{\otimes N})$, and
\begin{equation}\label{eq:optimal_parent}
    \rmM_\caH
    =\frac{1}{2}\sum_{a=1}^{d^2}\Gamma_a^\dag\Gamma_a
    = Z_\beta\,\Pi_\beta.
\end{equation}
Consequently $\operatorname{spec}(\rmM_\caH)=\{0,Z_\beta,\ldots,Z_\beta\}$,
\begin{equation}
    \mathrm{Gap}(\rmM_\caH)=\|\rmM_\caH\|=Z_\beta,
\end{equation}
so the maximal relative spectral gap is $1$ for
every finite-dimensional $H$ and every $\beta>0$.
\end{theorem}

\begin{proof}

Note that
\begin{equation}
    \Delta^{-1}(J_a)=e^{-\beta H/2} W_a,
    \qquad
    \Delta(J_a)=W_a e^{-\beta H/2},
    \qquad
    \Gamma_a = e^{-\beta H/2} W_a\otimes I - I\otimes (W_a e^{-\beta H/2})^\top,
\end{equation}
Completeness of basis gives
\begin{equation}\label{eq:completeness}
    \sum_{a=1}^{d^2}W_a(\cdot)W_a=\tr(\cdot)I,
    \qquad \sum_{a=1}^{d^2} W_a\otimes W_a^\top = |\Omega\>\<\Omega|,
\end{equation}
where $|\Omega\> \coloneqq \sum_{i=1}^{d}|ii\>$. Let $E_\beta \coloneqq e^{-\beta H/2}$. Then
\begin{equation}
    \rmM_\caH = \tr(E_\beta^2)I\otimes I - \frac{1}{2}(I\otimes E_\beta^\top)|\Omega\>\<\Omega|(E_\beta\otimes I) - \frac{1}{2}(E_\beta\otimes I)|\Omega\>\<\Omega|(I\otimes E_\beta^\top).
\end{equation}
Hence, the proof is completed by observing that
\begin{equation}
    Z_\beta^{1/2}|\rb^{1/2}\> = E_\beta \otimes I|\Omega\> = I\otimes E_\beta^\top|\Omega\>.
\end{equation}

\end{proof}

\begin{proposition}[Restatement of \pref{prop:replacer}]
For the generators $\{J_a = e^{-\beta H/4}W_a e^{-\beta H/4}\}_{a=1}^{d^2}$, condition
\eref{eq:lindblad_condition} holds, and
\begin{equation}
    \caL_{\rm SoS}(X)=Z_\beta\left(\Tr(X)\rb - X\right).
\end{equation}
\end{proposition}

\begin{proof}
Note that
\begin{equation}
    \sum_{a=1}^{d^2}J_a\otimes J_a^\top = \left[e^{-\beta H/4}\otimes \left(e^{-\beta H/4}\right)^\top\right]|\Omega\>\<\Omega|\left[e^{-\beta H/4}\otimes \left(e^{-\beta H/4}\right)^\top\right].
\end{equation}
A direct computation shows it commute with $\rb^{1/2}\otimes \rb^{-1/2,\top}$. Hence, $\caL_{\rm SoS}$ is a Lindbladian. 

To solve the Lindbladian, we first write $\rmM_\caH$ as
\begin{equation}
    \rmM_\caH = Z_\beta I\otimes I - \sum_{a=1}^{d^2} (E_\beta W_a) \otimes (E_\beta W_a)^\top.
\end{equation}
Therefore,
\begin{gather}
    \caH(X) = Z_\beta X - \sum_{a=1}^{d^2} E_\beta W_a X E_\beta W_a = Z_\beta\left[X - \tr\left(\rb^{\frac{1}{2}}X\right)\rb^{\frac{1}{2}}\right],\\
    \caL_{\rm SoS}(X) = -\rb^{\frac{1}{4}}\caH\left(\rb^{-\frac{1}{4}}X \rb^{-\frac{1}{4}}\right)\rb^{\frac{1}{4}} = \tr(X) e^{-\beta H} - Z_\beta X.
\end{gather}
\end{proof}

\section{Proof of \thref{thm:freefermion}}
\label{app:proof1}

Start with the DLL Lindbladian for free-fermion models \cite{ding2025efficient}. In the construction, the full Lindbladian  is
\begin{equation}
    \caL_{\mathrm{DLL}}(\cdot) = -\rmi[C,\cdot] + \sum_a \left(L_a (\cdot) L_a^\dag  - \frac{1}{2}\{L_a^\dag L_a,\cdot\}\right)
\end{equation}
where $\Delta^2(L_a) = L_a^\dag$ for all $a$. In particular, if $\left[\sum_a L_a^\dag L_a, H\right] = 0$, then $C = 0$, and the parent Hamiltonian is simply
\begin{equation}
    \rmM_{\caH,\mathrm{DLL}} = K \otimes I + I \otimes K^\top - \sum_a \Delta(L_a) \otimes \Delta(L_a)^*,\qquad K = \frac{1}{2}\sum_a L_a^\dag L_a.
\end{equation}
Let $H = \sum_{a,b}h_{ab}\omega_a \omega_b$ be a free-fermion model on $N$-qubit Hilbert space. We choose the set of jump operators constructed by Gaussian-filter function, as proposed by Refs. \cite{ding2025efficient,Smid2025}:
\begin{equation}
\begin{aligned}
    L_a &= \int_{-\infty}^{\infty} f(t) e^{\rmi H t}\omega_a e^{-\rmi H t} dt\quad a = 1,\ldots,2N,\\
    f(t) &= \sqrt{\frac{2}{\pi\beta^2}}\exp\left[-\frac{2}{\beta^2}\left(t - \rmi\frac{\beta}{4}\right)^2\right],\\
    \hat{f}(\nu) &= \int_{-\infty}^{\infty} f(t) e^{\rmi \nu t}dt = \exp\left[-\frac{(1 + \beta\nu)^2}{8} + \frac{1}{8}\right],
\end{aligned}
\end{equation}
and the coherent term is $0$. After computation, we obtain
\begin{equation}
    L_a = \sum_{b=1}^{2N}\left(e^{\beta h - 2\beta^2 h^2}\right)_{ab}\omega_b.
\end{equation}
Note that $(e^{-2\beta^2 h^2})^\top = e^{-2\beta^2 h^2}$, and $(e^{\beta h})^\top = e^{-\beta h}$. Hence, the Lindbladian is
\begin{equation}
    \caL_{\rm DLL}(\cdot) = \sum_{a,b=1}^{2N} \left(e^{-2\beta h -4\beta^2 h^2} \right)_{ab}\left(\omega_a (\cdot) \omega_b - \frac{1}{2}\{\omega_b \omega_a, \cdot\}\right).
\end{equation}
Observe that $\Delta^{\pm 1}(\omega_a) = \sum_{b=1}^{2N}(e^{\mp \beta h})_{ab}\omega_b$. Then the parent Hamiltonian of DLL Lindbladian is
\begin{gather}\label{eq:parentdll}
    \rmM_{\caH,\mathrm{DLL}} = K \otimes I + I \otimes K^\top - \sum_{a,b=1}^{2N}\left(e^{-4\beta^2 h^2}\right)_{ab}\omega_a \otimes\omega_b^\top,\\
    K = \frac{1}{2}\sum_{a,b=1}^{2N}\left(e^{2\beta h - 4\beta^2 h^2}\right)_{ab}\omega_a \omega_b.
\end{gather}
We first demonstrate that $\rmM_{\caH,\mathrm{DLL}}$ is an SoS parent Hamiltonian, and then show that there exists an entire family of such Hamiltonians parameterized by a coefficient matrix $S$.

\begin{theorem}[Restatement of \thref{thm:freefermion}]
Let
\begin{equation}
    J_a = \sum_{b=1}^{2N}(e^{-2\beta^2 h^2})_{ab}\omega_b\qquad \forall a = 1,\ldots,2N.
\end{equation}
Then $\mathrm{Alg}(\{J_a\}_{a=1}^{2N})$ is a set of unital irreducible algebra, and
    \begin{equation} 
    \mathrm{M}_{\caH,\mathrm{DLL}} = \frac{1}{2}\sum_{a=1}^{2N} \Gamma_a^\dag \Gamma_a,\qquad
    \Gamma_a = \Delta^{-1}(J_a) \otimes I - I \otimes \Delta(J_a)^\top.
\end{equation}
\end{theorem}
\begin{proof}
Since $\mathrm{Alg}(\{\omega_a\}_{a=1}^{2N})$ is a set of unital irreducible algebra, and $(e^{-2\beta^2 h^2})$ is invertible, the algebra $\mathrm{Alg}(\{J_a\}_{a=1}^{2N})$ is also irreducible.

Then we prove the second result. Note that $\rmM_{\caH,\rm DLL}$ has been derived in \eref{eq:parentdll}.
By definition, the modular annihilators are
\begin{gather}
    \Gamma_a = \sum_{b=1}^{2N}\left[\left(e^{\beta h - 2\beta^2 h^2}\right)_{ab} \omega_b \otimes I - \left(e^{-\beta h-2\beta^2 h^2}\right)_{ab} I \otimes \omega_b^\top\right]\qquad a = 1,\ldots,2N.
\end{gather}
Because $h^\top = h^* = -h$, 
\begin{equation}
    \Gamma_a^\dag = \sum_{b=1}^{2N}\left[\left(e^{-\beta h - 2\beta^2 h^2}\right)_{ab} \omega_b \otimes I - \left(e^{\beta h-2\beta^2 h^2}\right)_{ab} I \otimes \omega_b^\top\right]\qquad a = 1,\ldots,2N.
\end{equation}
After computation,
\begin{equation}\label{eq:middle}
     \frac{1}{2}\sum_{a=1}^{2N}\Gamma_a^\dag \Gamma_a = \frac{1}{2}\sum_{b_1,b_2}\left(e^{2\beta h - 4\beta^2 h^2}\right)_{b_1,b_2} (\omega_{b_1} \omega_{b_2} \otimes I + I \otimes \omega_{b_2}^\top \omega_{b_1}^\top) - \sum_{b_1,b_2=1}^{2N}\left(e^{-4\beta^2 h^2}\right)_{b_1,b_2}\omega_{b_1} \otimes \omega_{b_2}^\top.
\end{equation}
Here we have used the fact that
\begin{gather}
    \sum_{a,b_1,b_2}(e^{\beta h-2\beta^2 h^2})_{ab_1}(e^{-\beta h - 2\beta^2 h^2})_{ab_2}\omega_{b_1}^\top \omega_{b_2}^\top = \sum_{b_1,b_2}\left(e^{2\beta h - 4\beta^2 h^2}\right)_{b_1,b_2} \omega_{b_2}^\top \omega_{b_1}^\top,\\
    \sum_{a,b_1,b_2=1}^{2N}\left(e^{-\beta h - 2\beta^2 h^2}\right)_{ab_1}  \left(e^{-\beta h - 2\beta^2 h^2}\right)_{ab_2}\omega_{b_1} \otimes \omega_{b_2}^\top = \sum_{b_1,b_2=1}^{2N}\left(e^{-4\beta^2 h^2}\right)_{b_1,b_2}\omega_{b_1} \otimes \omega_{b_2}^\top.
\end{gather}
The proof is completed by comparing \eref{eq:middle} to \eref{eq:parentdll}.
\end{proof}

\section{Proof of \pref{prop:projector}}
\label{app:AGSP}

To analyze $\Gamma_a$, we introduce
    \begin{equation}
        \gamma_a \coloneqq \sum_{b=1}^{2N} \left(S e^{\beta h}\right)_{ab}\omega_b.
    \end{equation}
    Then 
    \begin{equation}
        \Gamma_a = \Delta^{-1}(J_a) \otimes I - I \otimes \Delta(J_a)^\top = \gamma_a \otimes I - I \otimes \gamma_a^*.
    \end{equation}
Note that using the condition $h^\top = -h$, we get
\begin{align}
    \gamma_a^2 &= \sum_{b_1,b_2} \left(S e^{\beta h}\right)_{ab_1}\omega_{b_1} \left(S e^{ \beta h}\right)_{ab_2}\omega_{b_2} = \sum_b \left(S e^{\beta h}\right)_{ab} ^2 = \left(S^2\right)_{aa}.
\end{align}
The result for $(\gamma_a^\dag)^2$ is the same. Therefore, $\gamma_a^2 = (\gamma_a^\dag)^2$ for all $a$.

Now we can complete the proof.
\begin{proposition}[Restatement of \pref{prop:projector}]
    Define
    \begin{equation}
        R_{+,a}
        \coloneqq
        \frac{\gamma_a+\gamma_a^\dagger}{2},
        \qquad
        R_{-,a}
        \coloneqq
        \frac{\gamma_a-\gamma_a^\dagger}{2\rmi}.
    \end{equation}
    Then $\{R_{+,a},R_{-,a}\}=0$, and there exist nonnegative scalars $r_{+,a}$ and $r_{-,a}$ such that
    \begin{equation}
        R_{+,a}^2=r_{+,a}I,
        \qquad
        R_{-,a}^2=r_{-,a}I.
    \end{equation}
    Furthermore,
    \begin{equation}
        \left(\Gamma_a^\dagger\Gamma_a\right)^2
        =
        4(r_{+,a} + r_{-,a})\,\Gamma_a^\dagger\Gamma_a.
    \end{equation}
\end{proposition}

\begin{proof}
    We omit the subscript $a$. Write
    \begin{equation}
        \gamma
        =
        \sum_{b=1}^{2N}(x_b+\rmi y_b)\omega_b
        =
        R_+ +\rmi R_-,
        \qquad
        x_b,y_b\in\bbR,
    \end{equation}
    where
    \begin{equation}
        R_+\coloneqq\frac{\gamma+\gamma^\dagger}{2}
        =\sum_{b=1}^{2N}x_b\omega_b,
        \qquad
        R_-\coloneqq\frac{\gamma-\gamma^\dagger}{2\rmi}
        =\sum_{b=1}^{2N}y_b\omega_b.
    \end{equation}
    Using the Majorana relations $\{\omega_b,\omega_c\}=2\delta_{bc}I$,
    we obtain
    \begin{equation}
        R_+^2=r_+ I,
        \qquad
        R_-^2=r_- I,
        \qquad
        \{R_+,R_-\}
        =
        2\left(\sum_{b=1}^{2N}x_by_b\right)I,
    \end{equation}
    where
    \begin{equation}
        r_+\coloneqq\sum_{b=1}^{2N}x_b^2,
        \qquad
        r_-\coloneqq\sum_{b=1}^{2N}y_b^2.
    \end{equation}
    On the other hand,
    \begin{equation}
        \gamma^2
        =
        R_+^2-R_-^2+\rmi\{R_+,R_-\},
        \qquad
        (\gamma^\dagger)^2
        =
        R_+^2-R_-^2-\rmi\{R_+,R_-\}.
    \end{equation}
    Since $\gamma^2=(\gamma^\dagger)^2$, it follows that
       $ \{R_+,R_-\}=0$.
    Define
    \begin{equation}
        R_z\coloneqq\frac{\rmi}{2}[R_+,R_-]=\rmi R_+ R_-
    \end{equation}
    which satisfies
    \begin{equation}
        R_z^2=r_+r_- I,
        \qquad
        \{R_z,R_+\}=\{R_z,R_-\}=0.
    \end{equation}
    Recall that
    \begin{equation}
        \Gamma
        =
        \gamma\otimes I-I\otimes\gamma^*
        =
        \left(R_+\otimes I-I\otimes R_+^\top\right)
        +
        \rmi\left(R_-\otimes I+I\otimes R_-^\top\right).
    \end{equation}
    A direct calculation gives
    \begin{equation}
        \Gamma^\dagger\Gamma
        =
        2\left[(r_+ + r_-)I+Q\right],
    \end{equation}
    where
    \begin{equation}
        Q
        \coloneqq
        -R_+\otimes R_+^\top
        +R_-\otimes R_-^\top
        +R_z\otimes I
        +I\otimes R_z^\top.
    \end{equation}
    After computation, it can be verified that
    \begin{equation}
        Q^2
        =
        (r_++r_-)^2I.
    \end{equation}
    Finally,
    \begin{align}
        \left(\Gamma^\dagger\Gamma\right)^2 =
        4[(r_+ + r_-)I+Q]^2 =
        8(r_+ + r_-)[(r_+ + r_-)I+Q] =
        4(r_+ + r_-)\,\Gamma^\dagger\Gamma.
    \end{align}
\end{proof}

\section{Proof of \pref{prop:cost}}
\label{app:proof_of_gap}

\begin{proposition}[Restatement of \pref{prop:cost}]
Suppose
   $ H=\sum_{a,b=1}^{2N}h_{ab}\omega_a\omega_b$
is a free-fermion Hamiltonian, where $h$ has spectrum
$\{\pm\lambda_n/2:n=1,\ldots,N\}$. Let $S=f(h)$, where $f$ is a
real-valued function satisfying
    $f(\lambda_n/2)=f(-\lambda_n/2)\neq 0$
for every $n$. Then
\begin{equation}
\begin{aligned}
    \operatorname{Gap}(\rmM_\caH)
    &=
    2\min_n
    \left\{
        f\left(\frac{\lambda_n}{2}\right)^2
        \cosh(\beta\lambda_n)
    \right\},\\
    \|\rmM_\caH\|
    &=
    4\sum_{n=1}^N
    f\left(\frac{\lambda_n}{2}\right)^2
    \cosh(\beta\lambda_n).
\end{aligned}
\end{equation}
\end{proposition}

\begin{proof}
Suppose coefficient matrix is 
\begin{equation}
    h = \bigoplus_{n=1}^N \frac{\rmi}{2}\begin{pmatrix}
        0 & \lambda_n \\
        -\lambda_n & 0
    \end{pmatrix}.
\end{equation}
Then $H = \rmi\sum_{n=1}^{N}\lambda_n \omega_{2n-1} \omega_{2n} = \sum_{n=1}^N H_n$, which is the sum of $N$ commuting blocks. Hence, 
\begin{equation}
    S = \bigoplus_{n=1}^N \begin{pmatrix}
        f_n & 0 \\
        0 & f_n
    \end{pmatrix},\qquad f_n \coloneqq f\left(\frac{\lambda_n}{2}\right) = f\left(-\frac{\lambda_n}{2}\right).
\end{equation}
Start with the first block $n=1$. By definition, we have
\begin{equation}
    J_1 = f_1 \omega_1 ,\quad J_2 = f_1 \omega_2.
\end{equation}
Note that $[H,\omega_1] = [H_1,\omega_1] =-\rmi 2\lambda_1 \omega_2, [H,\omega_2] = [H_1,\omega_2] = \rmi 2\lambda_1 \omega_1$. Hence, we can write $\rmM_\caH$ as summation of mutually commuting Hermitian operators:
\begin{equation}
    \frac{1}{2}\Gamma_{2n-1}^\dag \Gamma_{2n-1} + \frac{1}{2}\Gamma_{2n}^\dag \Gamma_{2n}\qquad n = 1,\ldots, N.
\end{equation}
For each $n$, the spectrum of $(\Gamma_{2n-1}^\dag \Gamma_{2n-1} + \Gamma_{2n}^\dag \Gamma_{2n})/2$ is
\begin{equation}
    \left\{0,\quad 2f_n^2 \cosh(\beta\lambda_n),\quad 2f_n^2 \cosh(\beta\lambda_n),\quad 4f_n^2 \cosh(\beta \lambda_n)\right\}.
\end{equation}
Therefore, 
\begin{equation}
    \mathrm{Gap}(\rmM_\caH) = 2\min_n \left[f_n^2\cosh(\beta \lambda_n)\right],\qquad \|\rmM_\caH\| = 4\sum_{n=1}^N f_n^2\cosh(\beta \lambda_n).
\end{equation}
In the more general scenarios where $H = \sum_{a,b=1}^{2N}h_{ab}\omega_a \omega_b$, there exists an orthogonal matrix $R$ such that
\begin{equation}
    R^\top h R = \bigoplus_{n=1}^N \frac{\rmi}{2}\begin{pmatrix}
        0 & \lambda_n \\
        - \lambda_n & 0
    \end{pmatrix}.
\end{equation}
Let $\tilde{\omega}_a \coloneqq \sum_b R_{ba} \omega_b$, then
\begin{equation}
    H = \rmi \sum_{n=1}^N \lambda_n \tilde{\omega}_{2n-1} \tilde{\omega}_{2n},
\end{equation}
and $\{\tilde{\omega}_a\}_{a=1}^{2N}$ is also a set of Majorana operators. If we introduce
\begin{equation}
    \tilde{J}_a = \sum_b (R^\top SR)_{ab}\tilde{\omega}_b,\quad 
    \tilde{\Gamma}_a = \Delta^{-1}(\tilde{J}_a) \otimes I - I \otimes \Delta(\tilde{J}_a)^\top,\quad \tilde{\rmM}_\caH = \frac{1}{2}\sum_a \tilde{\Gamma}_a^\dag \tilde{\Gamma}_a,
\end{equation}
Then the spectrum of $\tilde{\rmM}_\caH$ is
\begin{equation}
    \left\{\sum_{n=1}^N E_n : E_n \in \{0, 2f_n^2\cosh(\beta\lambda_n), 2f_n^2\cosh(\beta\lambda_n), 4f^2_n\cosh(\beta\lambda_n)\}\right\}.
\end{equation}
Observe that
\begin{equation}
    \tilde{J}_a = \sum_b R_{ba}J_b,\qquad \tilde{\Gamma}_a = \sum_b R_{ba}\Gamma_b.
\end{equation}
Due to the orthogonality of $R$, we have 
\begin{equation}
    \rmM_\caH = \frac{1}{2}\sum_a \Gamma_a^\dag \Gamma_a = \frac{1}{2}\sum_a \tilde{\Gamma}_a^\dag \tilde{\Gamma}_a = \tilde{\rmM}_\caH.
\end{equation}
This completes the proof of the proposition.
\end{proof}

\section{Proof of \pref{prop:mixing_time}}
\label{app:proof_of_mixing}

\begin{proposition}[Restatement of \pref{prop:mixing_time}]
Suppose $H = \sum_{a,b=1}^{2N}h_{ab}\omega_a \omega_b$ is a free-fermion model, where $h$ has spectrum $\{\pm \lambda_n/2 : n=1,\ldots,N\}$. Let $S=f(h)$, where $f(x)$ is a real-valued function satisfying $f(\lambda_n/2) = f(-\lambda_n/2)$ for all $n$. Then for every initial state $\rho_0$, 
\begin{equation}
\label{eq:worst_case_ff_simple}
    \frac12
    \left\|
    e^{t\mathcal L_{\rm SoS}}(\rho_0)-\rho_\beta
    \right\|_1
    \le
    2N e^{-g_\star t},\quad g_{\star} \coloneqq \min_n 2f(\lambda_n/2)^2\cosh(\beta\lambda_n). 
\end{equation}
Therefore,
\begin{equation}
    t_{\rm mix}(\epsilon)
    \le
    \frac{1}{g_\star}
    \log\left(\frac{2N}{\epsilon}\right).
\end{equation}
\end{proposition}

\begin{proof}
After diagonalising the single-particle Hamiltonian $h$, introduce
canonical fermionic modes
\begin{equation}
    \tilde{a}_n
    \coloneqq
    \frac12
    \left(
    \widetilde{\omega}_{2n-1}
    -
    i\widetilde{\omega}_{2n}
    \right).
\end{equation}
The free-fermion Lindbladian decomposes as
    $\mathcal L_{\rm SoS}
    =
    \sum_{n=1}^N\mathcal L_n$,
where
\begin{equation}
\label{eq:single_mode_GAD}
    \mathcal L_n(\cdot)
    =
    \gamma_{n,-}\left(\tilde{a}_n (\cdot)\tilde{a}_n^\dag  - \frac{1}{2}\{\tilde{a}_n^\dag \tilde{a}_n,\cdot\}\right)
    +
    \gamma_{n,+}\left(\tilde{a}_n^\dag (\cdot)\tilde{a}_n  - \frac{1}{2}\{\tilde{a}_n \tilde{a}_n^\dag,\cdot\}\right),
\end{equation}
with
\begin{equation}
    \gamma_{n,-}
    \coloneqq
    2f_n^2e^{-\beta\lambda_n},
    \qquad
    \gamma_{n,+}
    \coloneqq
    2f_n^2e^{\beta\lambda_n}.
\end{equation}
Let
\begin{equation}
    g_n
    \coloneqq
    \frac{\gamma_{n,+}+\gamma_{n,-}}{2}
    =
    2f_n^2\cosh(\beta\lambda_n).
\end{equation}

Note that
\begin{equation}
    [\mathcal L_n,\mathcal L_m]=0,
    \qquad \forall n\neq m.
\end{equation}
Consequently,
\begin{equation}
    e^{t\mathcal L_{\rm SoS}}
    =
    \prod_{n=1}^N e^{t\mathcal L_n}.
\end{equation}

Let $\caE_n\coloneqq\lim_{t\to\infty}e^{t\mathcal L_n}$.
Each $\mathcal L_n$ is unitarily equivalent to a single-qubit
generalized amplitude-damping generator tensored with the identity
on all other modes. Therefore it suffices to bound the corresponding
single-qubit channel.

Let $\tau_n \coloneqq(I+m_nZ)/2$ denote the stationary state of the $n$th mode, where
$|m_n|\le1$.  For the replacement channel
    $\mathcal R_{\tau_n}(\cdot) = \tr(\cdot)\tau_n$,
define
    $\Delta_{n,t}
    \coloneqq e^{t\caL_n} - \mathcal R_{\tau_n}$.
Its action on the Pauli basis is
\begin{gather}
    \Delta_{n,t}(I)
    =
    -m_ne^{-2g_nt}Z,
    \qquad
    \Delta_{n,t}(X)
    =
    e^{-g_nt}X,
    \\
    \Delta_{n,t}(Y)
    =
    e^{-g_nt}Y,
    \qquad 
    \Delta_{n,t}(Z)
    =
    e^{-2g_nt}Z.
\end{gather}
Thus, 
\begin{align}
\label{eq:delta_rank_one}
    \Delta_{n,t}(\cdot)
    &=
    \frac{e^{-g_nt}}{2}
    \tr(X\cdot)X
    +
    \frac{e^{-g_nt}}{2}
    \tr(Y\cdot)Y
    +
    \frac{e^{-2g_nt}}{2}
    \tr[(Z-m_nI)\cdot]Z,\\
    \|\Delta_{n,t}\|_\diamond 
    &\le
    2e^{-g_nt}
    +
    (1+|m_n|)e^{-2g_nt}
    \le
    2e^{-g_nt}
    +
    2e^{-2g_nt}.
\end{align}
Unitary invariance of the diamond norm therefore gives
\begin{equation}
    \|e^{t\caL_n}-\caE_n\|_\diamond
    \le
    2e^{-g_nt}
    +
    2e^{-2g_nt}.
\end{equation}
Using the telescoping identity,
\begin{align}
    \prod_{n=1}^N e^{t\caL_n}
    -
    \prod_{n=1}^N \caE_n
    =
    \sum_{n=1}^N
    \left(
    \prod_{j<n}e^{t\caL_j}
    \right)
    (e^{t\caL_n}-\caE_n)
    \left(
    \prod_{j>n}\caE_j
    \right),
\end{align}
we obtain
\begin{align}
    \left\|
    e^{t\mathcal L_{\rm SoS}}
    -
    \prod_{n=1}^N \caE_n
    \right\|_\diamond
    \le
    \sum_{n=1}^N
    \|e^{t\caL_n}-\caE_n\|_\diamond
    \le
    2\sum_{n=1}^N
    \left(
    e^{-g_nt}+e^{-2g_nt}
    \right).
\end{align}
Note that
\begin{equation}
    \prod_{n=1}^N\caE_n
    =
    \mathcal R_{\rho_\beta},
    \qquad
    \mathcal R_{\rho_\beta}(\cdot)
    =
    \tr(\cdot)\rho_\beta.
\end{equation}
Therefore, for every density matrix $\rho_0$,
\begin{align}
    \frac12
    \left\|
    e^{t\mathcal L_{\rm SoS}}(\rho_0)
    -
    \rho_\beta
    \right\|_1
    \le
    \frac12
    \left\|
    e^{t\mathcal L_{\rm SoS}}
    -
    \mathcal R_{\rho_\beta}
    \right\|_\diamond
    \le
    \sum_{n=1}^N
    \left(
    e^{-g_nt}
    +
    e^{-2g_nt}
    \right).
\end{align}
Using $g_n\ge g_\star$, we obtain
\begin{equation}
    \frac12
    \left\|
    e^{t\mathcal L_{\rm SoS}}(\rho_0)
    -
    \rho_\beta
    \right\|_1
    \le
    2Ne^{-g_\star t}.
\end{equation}
Solving the right-hand side for a target accuracy $\epsilon$
completes the proof.
\end{proof}

\section{Proof of \pref{prop:spectral_gap}}
\label{app:proof_prop}
    \begin{proposition}[Restatement of \pref{prop:spectral_gap}]
        Suppose $\mathrm{Alg}(\{J_a\}_{a\in\caA})$ is an unital irreducible algebra on $(\bbC^2)^{\otimes N}$ and $\rmM_\caH$ is an SoS parent Hamiltonian of the form \eref{eq:modular_annihilator} generated from $\{J_a\}_{a\in\caA}$ and each $J_a$ is Hermitian. It has ground state $|\Psi\>$ and spectral gap $\mathrm{Gap}(\rmM_\caH) = g$. Let $\widetilde{\Delta^{\pm 1}}(J_a)$ be the approximation of $\Delta^{\pm 1}(J_a)$ with operator norm error bound $\|\widetilde{\Delta^{\pm 1}}(J_a) - \Delta^{\pm 1}(J_a)\| \le \varepsilon$. Let 
        \begin{equation}
            \widetilde{\rmM}_\caH \coloneqq \frac{1}{2}\sum_{a\in\caA}\left(\widetilde{\Delta^{-1}}(J_a) \otimes I - I\otimes \widetilde{\Delta}(J_a)^\top\right)^\dag \left(\widetilde{\Delta^{-1}}(J_a) \otimes I - I\otimes \widetilde{\Delta}(J_a)^\top\right).
        \end{equation}
        If 
        \begin{equation}
            \sum_{a\in \caA} (\|\Gamma_a\|\varepsilon + \varepsilon^2) \le \frac{g}{8},
        \end{equation}        
        then $\widetilde{\rmM}_\caH$ has a unique ground state $|\widetilde{\Psi}\>$ with spectral gap $\mathrm{Gap}(\widetilde{\rmM}_\caH) \ge g/2$, and
        \begin{equation}
            1 - |\<\Psi|\widetilde{\Psi}\>|^2 \le \frac{4|\caA|\varepsilon^2}{g}.
        \end{equation}
    \end{proposition}

\begin{proof}
Let $\delta_a
    \coloneqq
    \widetilde{\Gamma}_a-\Gamma_a$. Since the transpose preserves the operator norm, the assumed
approximation bounds imply
\begin{equation}
            \|\delta_a\| \le \|\widetilde{\Delta}(J_a) - \Delta(J_a)\| + \|\widetilde{\Delta^{-1}}(J_a) - \Delta^{-1}(J_a)\| \le 2\varepsilon.
        \end{equation}
Note that
\begin{align}
    \widetilde{\rmM}_\caH-\rmM_\caH =
    \frac{1}{2}\sum_{a\in\caA}
    \left(
        \Gamma_a^\dagger\delta_a
        +
        \delta_a^\dagger\Gamma_a
        +
        \delta_a^\dagger\delta_a
    \right),
\end{align}
hence
\begin{align}
    \|
        \widetilde{\rmM}_\caH-\rmM_\caH
    \|\le
    \frac{1}{2}\sum_{a\in\caA}
    \left(
        2\|\Gamma_a\|\|\delta_a\|
        +
        \|\delta_a\|^2
    \right) \le
    2\sum_{a\in\caA}
    \left(
        \|\Gamma_a\|\varepsilon+\varepsilon^2
    \right)
    \le \frac{g}{4}.
\end{align}

Let $\widetilde{\lambda}_0$ and
$\widetilde{\lambda}_1$ be the two smallest eigenvalues of
$\widetilde{\rmM}_\caH$. Weyl's inequality gives
\begin{equation}
    \widetilde{\lambda}_0
    \le
    \|
        \widetilde{\rmM}_\caH-\rmM_\caH
    \|,
    \qquad
    \widetilde{\lambda}_1
    \ge
    g-
    \|
        \widetilde{\rmM}_\caH-\rmM_\caH
    \|.
\end{equation}
Therefore,
\begin{equation}
    \operatorname{Gap}(\widetilde{\rmM}_\caH)
    =
    \widetilde{\lambda}_1-\widetilde{\lambda}_0
    \ge
    g-
    2\|
        \widetilde{\rmM}_\caH-\rmM_\caH
    \|
    \ge \frac{g}{2}.
\end{equation}

Since $\Gamma_a|\Psi\rangle=0$, we have
\begin{align}
    \langle\Psi|
        \widetilde{\rmM}_\caH
    |\Psi\rangle =
    \frac{1}{2}\sum_{a\in\caA}
    \langle\Psi|
        \delta_a^\dagger\delta_a
    |\Psi\rangle \le
    \frac{1}{2}\sum_{a\in\caA}\|\delta_a\|^2
    \le
    2|\caA|\varepsilon^2.
\end{align}
Let $\widetilde{g}
\coloneqq \operatorname{Gap}(\widetilde{\rmM}_\caH)$. Using the spectral
decomposition of $\widetilde{\rmM}_\caH$ and the condition that the ground state of
$\widetilde{\rmM}_\caH$ is unique, we obtain
\begin{equation}
    \langle\Psi|
        \widetilde{\rmM}_\caH
    |\Psi\rangle
    \ge
    \widetilde{g}
    \left(
        1-
        |\langle\Psi|\widetilde{\Psi}\rangle|^2
    \right).
\end{equation}
Finally,
\begin{equation}
    1-
    |\langle\Psi|\widetilde{\Psi}\rangle|^2
    \le
    \frac{2|\caA|\varepsilon^2}{\widetilde{g}}
    \le
    \frac{4|\caA|\varepsilon^2}{g}.
\end{equation}
\end{proof}

\section{Proof of \thref{thm:local}}
\label{app:local}

We begin with a finite-range consequence of Theorem~2.3 and
Sec.~2.2.1 of Ref.~\cite{perez2023locality}.

\begin{theorem}\label{thm:araki}
Let
\begin{equation}
    H=\sum_j h_j,
    \qquad
    \operatorname{diam}\!\left(\operatorname{supp}(h_j)\right)
    \le r,
    \qquad
    J\coloneqq
    \sup_{\ell}
    \sum_{j:\,\ell\in\operatorname{supp}(h_j)}
    \|h_j\|.
\end{equation}
Suppose that $H$ is a one-dimensional Hamiltonian and that $V$ is a
single-site Pauli operator supported at site $i$. Let
\begin{equation}
    \Lambda_R(i)\coloneqq
    \{j:\operatorname{dist}(i,j)\le R\},
    \qquad
    H_R\coloneqq
    \sum_{j:\,\operatorname{supp}(h_j)\subseteq\Lambda_R(i)}h_j.
\end{equation}
Then, for every $x>0$, $\bar r \coloneqq r+1$,
\begin{equation}\label{eq:araki_finite_range}
\begin{aligned}
&\left\|
e^{xH}Ve^{-xH}
-
e^{xH_R}Ve^{-xH_R}
\right\|                                                     \\
&\quad\le
\exp\left[
    2xJ+
    4xJ(r+1)^2 e^{1+4xJ(r+1)}
\right]
\frac{
    \left[
        4xJ(r+1)^2e^{1+4xJ(r+1)}
    \right]^{
        \lfloor R/\bar r\rfloor+1
    }
}{
    \left(\lfloor R/(r+1)\rfloor+1\right)!
}.
\end{aligned}
\end{equation}
Moreover, if $0<\varepsilon\le e^{-1}$ and
\begin{equation}
    4xJ(r+1)<1,
\end{equation}
then it suffices to choose
\begin{equation}\label{eq:araki_radius_scaling}
    R = 
    O\left(
        xJr^3
        +
        r\log(\varepsilon^{-1})
    \right)
\end{equation}
in order to make the right-hand side of
\eref{eq:araki_finite_range} at most $\varepsilon$.

\end{theorem}

\begin{proof}
Let $\Phi_X$ denote the interaction obtained by collecting all terms
with support $X$:
\begin{equation}
    \Phi_X
    \coloneqq
    \sum_{j:\,\operatorname{supp}(h_j)=X}h_j.
\end{equation}
In the notation of Ref.~\cite{perez2023locality}, the interaction
strength satisfies
\begin{equation}
    \Omega_0
    =
    \sup_{\ell}
    \sum_{X\ni\ell}\|\Phi_X\|
    \le J.
\end{equation}
Moreover, the condition
$\operatorname{diam}(\operatorname{supp}(h_j))\le r$ implies
\begin{equation}
    \Omega_k=0
    \qquad
    \text{for every }k\ge\bar r=r+1.
\end{equation}

Ref.~\cite{perez2023locality} defines the complex-time evolution by
\begin{equation}
    \Gamma_\Lambda^s(A)
    =
    e^{\rmi sH_\Lambda}Ae^{-\rmi sH_\Lambda}.
\end{equation}
We apply Theorem~2.3 therein with
\begin{equation}
    A=V,
    \qquad
    s=-\rmi x,
    \qquad
    |s|=x,
    \qquad
    \ell=R.
\end{equation}
Since $V$ is a single-site Pauli operator,
Theorem~2.3 in Ref.\cite{perez2023locality} gives
\begin{equation}\label{eq:araki_intermediate}
\begin{aligned}
&\left\|
e^{xH}Ve^{-xH}
-
e^{xH_R}Ve^{-xH_R}
\right\| \le
e^{2xJ}
\sum_{k=R+1}^{L}
e^{4xJk}\Omega_k^*(4x),
\end{aligned}
\end{equation}
where $L$ is chosen sufficiently large that the corresponding
finite-volume Hamiltonian equals $H$.

Applying the finite-range estimate of Sec.~2.2.1 of
Ref.~\cite{perez2023locality} with range parameter $\bar r$ yields
\begin{equation}
\sum_{k=R+1}^{L}
e^{4xJk}\Omega_k^*(4x)
\le
e^a\frac{a^q}{q!},
\end{equation}
where
\begin{equation}
    \bar r \coloneqq r+1,\qquad a\coloneqq
    4xJ\bar r^{\,2}e^{1+4xJ\bar r},
    \qquad
    q\coloneqq
    \left\lfloor\frac{R}{\bar r}\right\rfloor+1.
\end{equation}
Substituting this estimate into
\eref{eq:araki_intermediate} proves
\eref{eq:araki_finite_range}.

It remains to derive the sufficient choice of $R$. By the elementary
factorial bound
    $q!\ge\left(q/e\right)^q$,
the right-hand side of \eref{eq:araki_finite_range} satisfies
\begin{equation}
    e^{2xJ+a}\frac{a^q}{q!}
    \le
    e^{2xJ+a}
    \left(\frac{ea}{q}\right)^q.
\end{equation}
It is sufficient to choose
\begin{equation}
    q\ge\max\left\{2ea,
    \frac{
        2xJ+a+\log(\varepsilon^{-1})
    }{\log 2}\right\}.
\end{equation}
Indeed, the first condition gives $ea/q\le1/2$, while the second
ensures $e^{2xJ+a}2^{-q}\le\varepsilon$. Consequently,
\begin{equation}
    q
    =
    O\left(
        1+a+xJ+\log(\varepsilon^{-1})
    \right).
\end{equation}

If $4xJ\bar r<1$, then
\begin{equation}
    a
    =
    4xJ\bar r^{\,2}e^{1+4xJ\bar r}
    \le
    4e^2xJ\bar r^{\,2}.
\end{equation}
Since $\bar r\ge1$, we also have
$xJ\le xJ\bar r^{\,2}$. For
$0<\varepsilon\le e^{-1}$, the constant term is absorbed by
$\log(\varepsilon^{-1})$, and hence
\begin{equation}
    q
    =
    O\left(
        xJ\bar r^{\,2}
        +
        \log(\varepsilon^{-1})
    \right).
\end{equation}
Because
$q=\lfloor R/\bar r\rfloor+1$, this proves
\eref{eq:araki_radius_scaling}.
\end{proof}

\begin{remark}
The correspondence with the notation of
Ref.~\cite{perez2023locality} is
\begin{equation}
    A\mapsto V,
    \qquad
    s\mapsto-\rmi x,
    \qquad
    |s|\mapsto x,
    \qquad
    \ell\mapsto R,
    \qquad
    \Omega_0\mapsto J,
    \qquad
    |\operatorname{supp}(A)|\mapsto1.
\end{equation}
The auxiliary parameter $\bar r=r+1$ accounts for the difference
between the convention
$\operatorname{diam}(\operatorname{supp}(h_j))\le r$ used here and
the convention $\Omega_k=0$ for $k\ge\bar r$ used in
Sec.~2.2.1 of Ref.~\cite{perez2023locality}.
\end{remark}

Then we can prove \thref{thm:local} in the main text.
\begin{theorem}[Restatement of \thref{thm:local}]
Suppose $H$ is an $N$-qubit Hamiltonian of the form \eref{eq:finiterange}, and $V$ is a Pauli operator supported on the $i$th site, and $\beta < 1/(Jr+J)$. Then for error tolerance \(0 < \varepsilon < e^{-1}\), there exists a truncated Hamiltonian \(H_{R,i}\) supported on a neighborhood of the $i$th site of radius
\begin{equation}
    R = O(\beta Jr^3 + r\log(\varepsilon^{-1})),
\end{equation}
such that the \(m\) step Lanczos approximation satisfies
\begin{equation}
    \|\Delta(V)-\caK(m,\operatorname{ad}_{H_{R,i}},\beta/4,V)\| \le \varepsilon.
\end{equation}
provided that
\begin{equation}
m=O\left(R+ \log(\varepsilon^{-1})\right) = O(\beta J r^3 + r\log(\varepsilon^{-1})).
\end{equation}
\end{theorem}
\begin{proof}
Introduce the effective Frobenius norm $\|\cdot\|_{R,F}$ on subspace $\caH_R \coloneqq \bigotimes_{j\in\Lambda_R(i)}\caH_j$, where $\caH_j$ is the local Hilbert space on the $j$th site. Note that for any $X\in \caB(\caH_R)$, we have $\|X\| \le \|X\|_{R,F}$. This Hilbert space has dimension $2^{2R+1}$. By virtue of \thref{thm:worst}, with
\begin{equation}
    m = O(\|H_{R,i}\|\beta + \log(\varepsilon_{\rmK}^{-1}))
\end{equation}
steps of Lanczos recursion, one can ensure that
\begin{equation}
    \left\|\caK(m,\operatorname{ad}_{H_{R,i}},\beta/4,V) - e^{\beta H_{R,i}/4} V e^{-\beta H_{R,i}/4}\right\|_{R,F} \le \sqrt{2^{2R+1}}\varepsilon_K.
\end{equation}
In conjugation with the Araki-type bound in \thref{thm:araki}, when $\beta J(r+1) < 1$, and
\begin{equation}
    R = O\left(\beta Jr^3 + r \log(\varepsilon_{\rm LR}^{-1})\right),
\end{equation}
we have
\begin{align}
    \|\caK(m,\operatorname{ad}_{H_{R,i}},\beta/4,V) - \Delta(V)\| &\le  \left\|\caK(m,\operatorname{ad}_{H_{R,i}},\beta/4,V) - e^{\beta H_{R,i}/4} V e^{-\beta H_{R,i}/4}\right\|  \nonumber \\
    &\quad + \left\|e^{\beta H_{R,i}/4} V e^{-\beta H_{R,i}/4} - \Delta(V)\right\| \nonumber\\
    &\le 2^{R+1/2}\varepsilon_{\rm K} + \varepsilon_{\mathrm{LR}}.
\end{align}
The theorem is then proved by choosing $\varepsilon_{\mathrm{LR}} = \varepsilon/2, \varepsilon_\rmK = \varepsilon/2^{R+3/2}$ and using the upper bound $\|H_{R,i}\| \le (2R+1)J$.

\end{proof}

\end{document}